\documentclass[12pt]{article}
\usepackage{epsfig, amssymb, graphicx, amsmath} 
\usepackage{epstopdf}
\usepackage{amsthm}
\usepackage[round-mode=places, round-integer-to-decimal, round-precision=2,
table-format = 1.2, 
table-number-alignment=center,
round-integer-to-decimal,
output-decimal-marker={,}
]{siunitx} 
\usepackage[textsize=scriptsize,textwidth=2cm,shadow]{todonotes}
\usepackage{booktabs}
\usepackage{enumerate}
\usepackage{eurosym}	
\usepackage{wrapfig}	
\usepackage{float}		
\usepackage{subfig}	
\usepackage{footnote}
\usepackage{bbold}

\usepackage{amsthm}
\usepackage{thmtools}
\usepackage{bbm}		
\usepackage[toc,page]{appendix}
\usepackage{bigints}
\usepackage[authoryear]{natbib}
\bibpunct{(}{)}{,}{a}{}{,}

\newtheorem{theorem}{Theorem}

\newtheorem{proposition}[theorem]{Proposition}

\newtheorem{lemma}[theorem]{Lemma}

\theoremstyle{definition}
\newtheorem{definition}[theorem]{Definition}
\theoremstyle{remark}

\numberwithin{theorem}{section}
\numberwithin{proposition}{section}
\numberwithin{lemma}{section}
\numberwithin{corollary}{section}
\numberwithin{definition}{section}
\numberwithin{remark}{section}
\numberwithin{example}{section}

\newcommand{\be}{\begin{equation}}
\newcommand{\en}{\end{equation}}
\newcommand{\ben}{\begin{equation*}}
\newcommand{\enn}{\end{equation*}}
\newcommand{\bea}{\begin{eqnarray}}
\newcommand{\ena}{\end{eqnarray}}

\begin{document}
	
	\newlength\tindent
	\setlength{\tindent}{\parindent}
	\setlength{\parindent}{0pt}
	\renewcommand{\indent}{\hspace*{\tindent}}
	
	\begin{savenotes}
		\title{
			\bf{ 
				Short-time implied volatility of \\ additive normal tempered stable  processes \\
		}}
		\author{
			Michele Azzone$^{\ddagger\,\S}$ \& 
			Roberto Baviera$^\ddagger$ 
		}
		
		\maketitle
		
		\vspace*{0.11truein}
		\begin{tabular}{ll}
			$(\ddagger)$ &  Politecnico di Milano, Department of Mathematics, 32 p.zza L. da Vinci, Milano \\
			$(\S)$& European Central Bank\footnote{The views expressed are those of the author and do not necessarily reflect the views of ECB.}\\
		\end{tabular}
	\end{savenotes}
	
	\vspace*{0.11truein}
	\begin{abstract}
		\noindent 
		Empirical studies have emphasized that the equity implied volatility is characterized by a negative skew inversely proportional to the square root of the time-to-maturity.\\
		We examine the short-time-to-maturity behavior of the implied volatility smile
		for pure jump exponential additive processes. An excellent calibration of the equity volatility surfaces has been achieved by a class of these additive processes with power-law scaling. The two power-law scaling parameters are $\beta$, related to the variance of jumps, and $\delta$, related to the smile asymmetry. It has been observed, in option market data, that $\beta=1$ and $\delta=-1/2$.\\
		In this paper, we prove that the implied volatility of these additive processes is consistent, in the short-time,
		with the equity market empirical characteristics if and only if $\beta=1$ and $\delta=-1/2$.\\
		\noindent
		
	\end{abstract}
	
	\vspace*{0.11truein}
	\vspace*{0.11truein}
	
	
	\vspace{3cm}
	\begin{flushleft}
		{\bf Address for correspondence:}\\
		Roberto Baviera\\
		Department of Mathematics \\
		Politecnico di Milano\\
		32 p.zza Leonardo da Vinci \\ 
		I-20133 Milano, Italy \\
		Tel. +39-02-2399 4575\\
		Fax. +39-02-2399 4621\\
		roberto.baviera@polimi.it
	\end{flushleft}

	\newpage
	
	\begin{center}
		\Large\bfseries 
		Short-time implied volatility of \\ additive normal tempered stable  processes 
		
	\end{center}
	
	\vspace*{0.21truein}
	
\section{Introduction}

Which characteristics of the implied volatility surface  should be reproduced by an option pricing model?  
A \textit{stylized fact} that characterizes the equity market 
is a downward slope in terms of strike, i.e. a  negative skew, 
where the skew is the at-the-money (ATM) derivative of the implied volatility w.r.t. 
the moneyness.\footnote{The moneyness is the logarithm of the strike price over the forward price. 
For a description of the equity volatility surface and a definition of skew, see, e.g., \citet{gatheral2011volatility}.} 
Specifically, the short-time\footnote{We refer to short-time-to-maturity.} negative skew is proportionally inverse to the square root of the time-to-maturity. 
The first empirical study 
of the equity skew dates back to \citet{carr2003finite}: they find that the S\&P 500 short-time skew is, on average, asymptotic to $-0.25/\sqrt{t}$. 
\citet{fouque2004maturity} arrive at a similar conclusion considering only options with short-time-to-maturity (i.e. up to three months).  In this paper, we show that a pure jump additive process, which also calibrates accurately the whole equity volatility surface, reproduces the power scaling market skew. 

\bigskip
	
A vast literature on short-time implied volatility and skew is available for jump-diffusion processes. 
Both the ATM 
\citep[see, e.g.][]{alos2007short, roper2009implied,muhle2011small,AndersenLipton,figueroa2016high}
and the OTM implied volatility 
\citep[see, e.g.][]{tankov2011pricing,figueroa2012small,mijatovic2016new,figueroa2018short}
are analyzed.
For a jump-diffusion L\'evy process,
	the ATM implied volatility is determined uniquely by the diffusion term; it goes to zero as the time-to-maturity goes to zero if 
	there is no diffusion term, i.e. for a pure-jump process. 
For this reason, pure jumps L\'evy processes are not suitable to reproduce the market short-time smile, 
because the short-time implied volatility is strictly positive in 
all financial markets. 

\citet{muhle2011small} have shown that, 
for a relatively broad class of additive models, 
the ATM behavior at small-time
is the same as the corresponding Levy.
In this paper, we analyze the ATM implied volatility and skew for a class of pure jump additive processes that is consistent with the equity market smile, differently from the L\'evy case: this is the main theoretical contribution of this study.
	
\bigskip
	
An additive process is a stochastic process with independent but non-stationary increments;
	a detailed description of the main features of additive processes is provided by \citet{Sato}. 
	In this paper, we focus on a pure jump additive extension of the well-known L\'evy normal tempered stable process \citep[for a comprehensive description 
of this set of  L\'evy processes, see, e.g.,][Ch.4]{Cont}.

\smallskip

Pure jump processes present a main 
advantage w.r.t.  jump-diffusion models: 
they generally describe underlying dynamics 
more  parsimoniously. 
In a jump-diffusion, both small jumps 
and the diffusion term describe little changes in the process \citep[see, e.g. ][]{Asmussen2001}.
Because both components of the jump-diffusion process are 
qualitatively similar, 
when calibrating the model to the plain vanilla option market, 
it is rather difficult to disentangle the two components 
and several sets of parameters 
achieve similar results.

\smallskip
Recently, it has been introduced a class of pure jump additive processes, the power-law scaling additive normal tempered stable process (hereinafter ATS),
where
the two key time-dependent parameters --the variance of jumps per unit of time, $k_t$, and the asymmetry  parameter, $\eta_t$--  
present a power scaling w.r.t. the time-to-maturity $t$.
 It has been shown the excellent calibrating performances of this class of processes \citep[see, e.g. ][]{azzone2019additive}. 
	On the one hand, 
this class of pure jump additive processes allows calibrating  
	the S\&P 500 and EURO STOXX 50 implied volatility surfaces with great accuracy, 
	reproducing
	``exactly" the term structure of the equity market implied volatility surfaces.
	On the other hand, the observed reproduction of the skew term structure appears remarkable. 

Moreover, an interesting self-similar characteristic w.r.t. the time-to-maturity arises. 
Specifically, among all allowed power laws, the power scaling of  $k_t$, $\beta$, is close to one, while the power scaling of $\eta_t$, $\delta$, is statistically
consistent with minus one half \citep[see, e.g. ][]{azzone2019additive}.

\bigskip	
	Consider an 
option price with strike $K$ and time-to-maturity $t$. We define $I_t(x)$ the model implied volatility,
where $x:=\log{\frac{K}{F_0(t)}}$ is the moneyness and $F_0(t)$ is the underlying forward price with time-to-maturity $t$.
In particular,  we consider the  \textit{moneyness degree}
$y$, s.t. $x=:y\sqrt{t}$, introduced by \citet{medvedev}. 
It has been observed that the  \textit{moneyness degree} $y$ can be interpreted as the distance of the option moneyness from the forward price in terms of the Black Brownian motion standard deviation \citep[see, e.g.,][]{carr2003finite,medvedev}. The implied volatility w.r.t  $y$ is 
\[
	{\cal I}_{t}(y) := I_t(y \sqrt{t})\;\;,
\] 
and its first order Taylor expansion w.r.t. $y$ in $y=0$  is
\[ 
	{\cal I}_{t}(y) = {\cal I}_{t}(0) +y\,\left.\frac{d {\cal {I}}_t(y)}{dy}\right\vert_{y=0}+ o(y)=:\hat{\sigma}_t+y\,\hat{\xi}_t+o(y)\;\;.
\]
We call $\hat{\xi}_t$ the skew term. We define $\hat{\sigma}_0$ and $\hat{\xi}_0$ as the the limits for $t$ that goes to zero of $\hat{\sigma}_t$ and $\hat{\xi}_t$.  
Their financial interpretation is straightforward: 
$\hat{\sigma}_0$ corresponds to the short-time ATM implied volatility, while  $\hat{\xi}_0$ is related to the short-time skew,
because it is possible to write  the skew as 
\[ 
\left.\frac{d {I}_t(x)}{dx}\right\vert_{x=0} = \frac{\hat{\xi}_t}{\sqrt{t}} \;\;.
\] 

	In Figure \ref{Figure::Short_time_skew}, we present an example of  the short-time implied volatility and the skew for the S\&P 500  at a given date,
 the $22^{nd}$ of June 2020 (the business day 
after a quadruple witching Friday\footnote{A quadruple witching Friday is the third Friday of the months of March, June, September and December: in this quarterly date,
stock options, stock futures, equity index futures, and equity index options all expire on the same day.}). 
On the left, we plot the one month (blue circles), two months (red squares), three months (orange stars), and four months (purple triangles) market implied volatility w.r.t. the \textit{moneyness degree} $y$: we observe a positive and finite short-time $\hat{\sigma}_t$. 
On the right, we plot the market skew w.r.t. the time $t$: it appears to be well described by a fit $O\left(\sqrt{\frac{1}{t}}\right)$.  
	\begin{center}
		\begin{minipage}[t]{1\textwidth}
			\includegraphics[width=\textwidth]{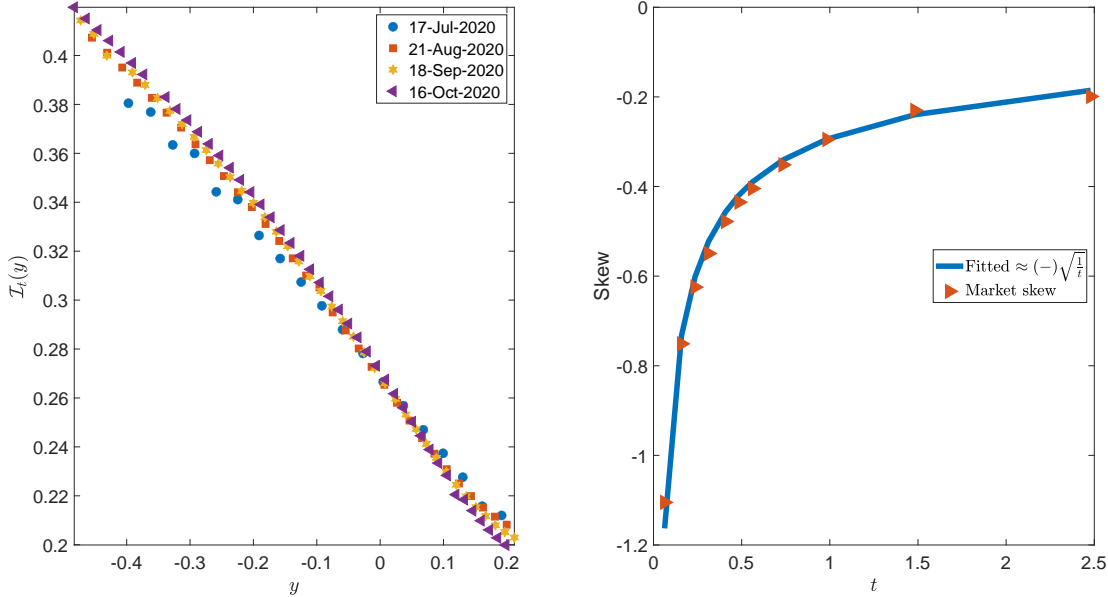}
\label{Figure::Short_time_skew}\captionof{figure}{\small Example of the S\&P 500 short-time implied volatility and  skew on the $22^{nd}$ of June 2020. On the left, we plot the one month (blue circles), two months (red squares), three months (orange stars), and four months (purple triangles) market implied volatility w.r.t. the \textit{moneyness degree} $y$. We observe a positive short-time $\hat{\sigma}_t$. On the right, we plot the market skew w.r.t. the time $t$ and the fitted $\approx \left(-\right)\sqrt{\frac{1}{t}}$.}
		\end{minipage} 
	\end{center}
As already observed in some empirical studies \citep[see, e.g.][]{carr2003finite,fouque2004maturity},
equity market data are compatible with a positive and finite $\hat{\sigma}_0$ and
a negative and finite $\hat{\xi}_0$, that  leads to a skew  proportionally inverse to the square root of the time-to-maturity.
We aim to present a pure-jump model with these features.

\smallskip

We study the behavior of $\hat{\sigma}_t$ and $\hat{\xi}_t$ for the ATS process, deriving, in (\ref{equation::call}, \ref{equation::put}),  an extension of the  \citet[][p.4, Eq.(7)]{hull1987pricing} formula
   \citep[see, e.g., ][for another application of this formula to the short-time case]{alos2007short}. 
This formula leads to two results:
on the one hand, we build some relevant bounds for $\hat{\sigma}_t$;   
on the other hand, we obtain an expression for $\hat{\xi}_t$ in (\ref{equation::smirk}) via  the implicit function theorem \citep[see, e.g.][Th.11, p.164]{loomis1968advanced}.\\

Three are the main contributions of this paper.
First, we deduce 
for a family of pure-jump additive processes, the  ATS, the behavior of the short-time ATM implied volatility $\hat{\sigma}_t$ and skew term $\hat{\xi}_t$  
over the region of admissible parameters (see \textbf{Theorem \ref{theorem:semplified_f}}).
Second, we prove that
only the scaling parameters observed in market data ($\beta=1$ and $\delta=-1/2$) are compatible with a finite short-time implied volatility and a short-time skew proportionally inverse to the square root of the time-to-maturity. 
Third, we demonstrate it exists a pure-jump additive process (an exponential ATS) that presents the two features observed in market data:
not only a finite and positive short-time implied volatility but also a power scaling skew.\\

	
The rest of the paper is organized as follows. 
Section 2 presents the ATS power scaling process and  the extension of the  Hull and White formula. 
Section 3 defines the implied volatility problem and analyzes the short-time ATM implied volatility $\hat{\sigma}_t$. 
Section 4 computes the short-time limit of the skew term $\hat{\xi}_t$. 
Section 5 presents the major result: the ATS process is consistent with the equity market if and only if $\beta=1$ and $\delta=-1/2$. Finally, Section 6 concludes. In the appendices, we report some technical lemmas: on basic properties in Appendix A and
on short-time limits in Appendix B.

\section{The ATS implied volatility}

In this Section, we recall the  characteristic function of the power-law scaling additive normal tempered stable process (ATS) and the notation employed in the paper.
We also introduce a sequence of random variables (\ref{eq: master}) with the same distribution of the ATS for any fixed time $t$; 
we use these random variables to study the short-time implied volatility. 

We discuss the volatility smile at small-maturity produced by this forward model and determine the power laws of the ATS parameters that are consistent with the market data i.e. which choices of $\beta$ and  $\delta$ reproduce the market short-time features mentioned above.\\

We define a sequence of positive random variables $S_t$ via its Laplace transform. 
The random variable $S_t$ appears in the definition of the random variable $f_t$, that is used to model a forward contract of the underlying of interest.
 
\begin{definition} {\bf Definition of $\left \{S_t\right \}_{t\geq 0}$} \\
	\label{definition::S_t}
	Let $\left \{S_t\right \}_{t\geq 0}$ be a sequence of positive random variables with a Laplace transform s.t.\\
	\( \ln {\cal L}_t \left(u;\;k_t,\;\alpha\right) := \ln \mathbb{E}\left[e^{-uS_t}\right]= 
	\begin{cases} 
	\displaystyle \frac{t}{k_t}
	\displaystyle \frac{1-\alpha}{\alpha}
	\left \{1-		\left(1+\frac{u \; k_t}{(1-\alpha)t}\right)^\alpha \right \} & \mbox{if } \; 0< \alpha < 1 \\[4mm]
	\displaystyle -\frac{t}{k_t}
	\ln \left(1+ \frac{u \;k_t}{t}\right)  & \mbox{if } \; \alpha = 0 \end{cases}\;\; , \; t\ge 0\)\\
	where $k_t:=\bar{k}t^\beta$ and $\bar{k},\;\beta \in \mathbb{R}^+$.
	
\end{definition}
Notice that, by the Laplace transform, we can compute any  moment of $S_t$. The first two are
\begin{enumerate}
	\item \(\mathbb{E}\left[S_t\right]=1\);
	\item \(Var\left[S_t\right]=k_t/t\).
\end{enumerate}

\begin{definition} {\bf Definition of $\left \{f_t\right \}_{t\geq 0}$} \\
	\label{definition::f_t}
	Let $\left \{f_t\right \}_{t\geq 0}$ be a sequence of random variables with characteristic function s.t.
	\begin{equation}
	\label{eq:characteristic}
	\mathbb{E}\left[e^{iuf_t}\right]:={\cal L}_t \left(iu t \left(\frac{1}{2}+\eta_t \right)\bar{\sigma}^2+t\frac{u^2\bar{\sigma}^2}{2};\;k_t,\;\alpha \right)e^{i\,u\,\varphi_t\,t}\;\;,  \; t\ge 0
	\end{equation}
where \begin{equation}
  \eta_t:=\bar{\eta}t^\delta\;\text{and}\;  \varphi_t\, t:=-\ln {\cal L}_t\left(t\bar{\sigma}^2 \eta_t;\;k_t,\;\alpha\right) \label{equation::power_scaling_param}\;\;.
\end{equation}
	 $\bar{\sigma},\,\bar{\eta} \, \in \mathbb{R}^+$ and $\delta\, \in \mathbb{R}.$
	
\end{definition}

Notice that the characteristic function is the same as the  power-law scaling additive normal tempered stable process (ATS) in \citet[][Eq.2]{azzone2019additive}; the notation has been slightly simplified, we report it at the end of the paper.
We also define $F_0(t)$, the forward contract at time $0$ with maturity $t$, and model  the same forward contract at maturity as
\[
 F_t(t) := F_0(t) \; e^{f_t} \;\; .
\]


We report the known result on the existence of an additive process with characteristic function (\ref{eq:characteristic}), cf. 
\citet[][Th.2.3]{azzone2019additive}.

\begin{theorem}{\bf Power-law scaling ATS\\}\label{theorem:semplified_f}
It exits an additive process with the same characteristic function of (\ref{eq:characteristic}),
where  $\alpha \in [0,1)$ and $\beta, \delta\in \mathbb{R}$ with either $\beta=\delta=0$ or
\begin{enumerate}
			\item $ \displaystyle 0 \leq \beta\leq \frac{1}{1-\alpha/2}\;\;;$ 
			\item $-\min\left(\beta, \dfrac{1-\beta\left(1-\alpha\right)}{\alpha} \right)<\delta\leq 0\;\;;$
\end{enumerate}
where the second condition reduces to $ -\beta< \delta \leq 0$ for $\alpha =0$ \qed
\end{theorem}
The region of admissible values for the scaling parameters $\beta$ and $\delta$ is shown in Figure \ref{Figure::Admissible}.

In particular, we mention that, $\forall \, \alpha \in [0,1)$,	
the scaling parameters observed in the market, $\left\{\delta=-1/2,\; \beta=1\right\}$, are always inside the ATS admissible region. 
	In Figure \ref{Figure::Admissible}, we plot the admissible region for the scaling parameters $\beta$ and $\delta$. In this paper, we prove that the ATS implied volatility at short-time is qualitatively different for different sets of scaling parameters. We separate the admissible region in five Cases:\\
    Case 1 (grey area): 
\begin{equation*}
        \left\{\beta<1, \;\; -\min\left(\frac{1}{2},\beta\right)<\delta\leq 0\right\} \cup \left\{\delta=\beta=0 \right\}\;\;.
    \end{equation*}
Case 2 (orange area): 
\begin{equation*}
\left\{-\min\left(\beta,\frac{1-\beta(1-\alpha)}{\alpha}\right)<\delta<-\frac{1}{2} \, \max\left({\beta},1\right)\right\}\;\;.    
\end{equation*}
 Case 3 (light green area): 
 \begin{equation*}
\left\{\beta\geq 1,\;\;     -\frac{\beta}{2}\leq\delta\leq 0\right\}\setminus \left\{\beta=1,\;\; \delta=-\frac{1}{2}\right\}\;\;.
 \end{equation*}
Case 4 (continuous dark green line): 
\[
\left\{\beta<1,\;\;\delta=-\frac{1}{2}\right\}\;\;.
\]
Case 5 (red dot): 
\begin{equation*}
	  \left\{\beta=1,\;\; \delta=-\frac{1}{2}\right\}\;\;.
\end{equation*}
	
Notice that Case 3 includes all its boundaries, identified by the green circles, with the exception of the  point  $\{ \beta=1,\;\delta=-1/2\}$ (red); 
Case 1 includes just its upper boundary (but it does not include its lower boundary), identified by the grey squares. 

The main objective of this paper is to prove that the $5$ Cases correspond to different behaviors of the implied volatility in the short-time. 
A summary of the ATS short-time behavior, and in particular of the ATM value $\hat{\sigma}_0$ and of the skew term $\hat{\xi}_0$, w.r.t. the different Cases is available in Table \ref{tab:Results}.

\begin{center}
	\begin{tabular}{c| l}
		Case 1 &  $\hat{\sigma}_0=0$ \\
		Case 2 &  $\hat{\sigma}_0=\infty$  \\
		Case 3 & $\hat{\sigma}_0>0$ \; and \; $\hat{\xi}_0=0$\\
		Case 4 & $\hat{\sigma}_0>0$ \; and \; $\hat{\xi}_0=-\sqrt{\frac{\pi}{2}}$\\
		Case 5 &  $\hat{\sigma}_0>0$ \; and \; $\hat{\xi}_0<0$\\		
	\end{tabular}	\captionof{table}{\small Summary of ATS short-time implied volatility behavior in the five Cases. We observe that only the last $3$ Cases admit a positive (and finite) implied volatility.
	}		\label{tab:Results}
\end{center}

\begin{center}
		\begin{minipage}[t]{1\textwidth}
			\includegraphics[width=\textwidth]{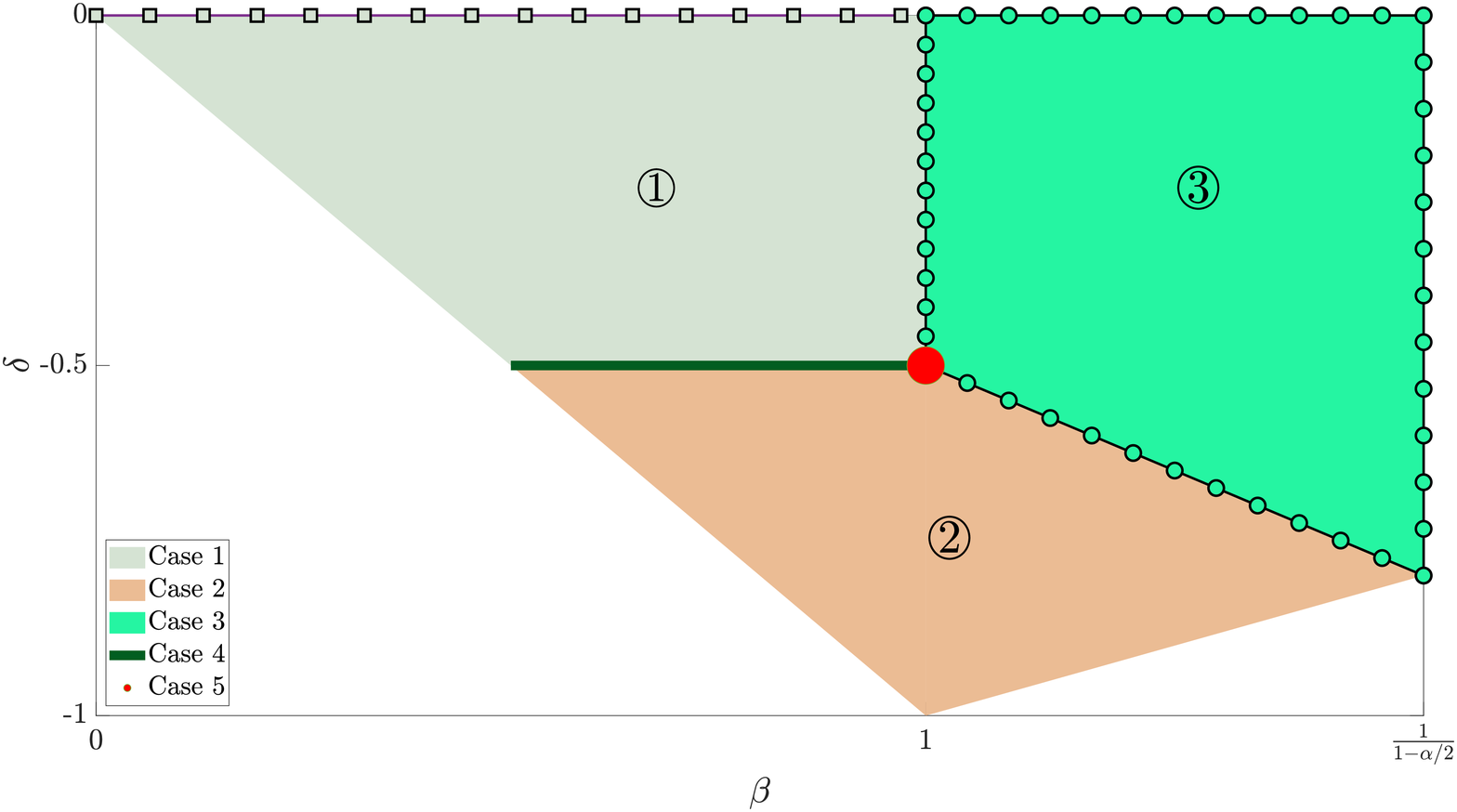}\label{Figure::Admissible}
\captionof{figure}{\small ATS admissible region for the scaling parameters. We separate the region in five Cases. \\i) Case
		 1 (grey area)  with  $\hat{\sigma}_0=0$. ii) Case 2 (orange area) with  $\hat{\sigma}_0=\infty$. iii) Case 3 (light green area) with finite $\hat{\sigma}_0$ and $\hat{\xi}_0=0$. iv) Case 4 (continuous dark green line) with finite $\hat{\sigma}_0$ and $\hat{\xi}_0=-\sqrt{\frac{\pi}{2}}$.
			v) Case 5 (red dot) with finite $\hat{\sigma}_0$ and negative and finite $\hat{\xi}_0$.\\
			Notice that Case 3 includes all its boundaries, identified by the green circles, with the exception of the point  $\{ \beta=1,\;\delta=-1/2\}$ (red), that corresponds to Case 5. Moreover, Case 1 includes just its upper bound, identified by the grey squares. We emphasize that for all $\alpha$ in $[0,1)$ the point $\{ \beta=1,\;\delta=-1/2\}$ is inside the admissible region.}
		\end{minipage} 
	\end{center}

It is also useful to provide the same result dividing the region for the admissible values of {\bf Theorem \ref{theorem:semplified_f}},  in terms  of $\beta$ and $\delta$.
A summary of the ATS short-time behavior, w.r.t. the scaling parameters $\beta$ and $\delta$   in the additive process admissible region is reported in Table \ref{tab:Results1}.
	\begin{center}
	\begin{tabular}{c| c c c}
		\bottomrule
		& $0<\beta<1$ & $\beta=1$& $\beta>1$\\
		\bottomrule
		$\delta>-\frac{1}{2}$ &$\hat{\sigma}_0=0$&$\hat{\xi}_0=0$&   $\hat{\xi}_0=0$ \\
		$\delta=-\frac{1}{2}$ &  $\hat{\xi}_0=-\sqrt{\frac{\pi}{2}}$&    $\boldsymbol{\hat{\xi}_0<0}$&$\hat{\xi}_0=0$ \\
		$\delta<-\frac{1}{2}$ & $\hat{\sigma}_0=\infty$&    $\hat{\sigma}_0=\infty$&  $\hat{\sigma}_0=\infty$     or $\hat{\xi}_0=0$ \\
		\bottomrule
		
	\end{tabular}	\captionof{table}{\small Summary of ATS short-time  implied volatility behavior, w.r.t. the scaling parameters $\beta$ and $\delta$   in the additive process admissible region. The ATM implied volatility $\sigma_0$  is positive (and finite) when we report the value for $\hat{\xi}_0$. 
}		\label{tab:Results1}
\end{center}

It can be proven that, for every time $t$, the random variable  
\begin{equation}
    \label{eq: master}
f_t =-  \left( \eta_t + \frac{1}{2} \right) \; \bar{\sigma}^2 \; S_t\;t  +  \; \bar{\sigma} \sqrt{S_t t} \, g  + \varphi_t t  \; 
\end{equation}
has the characteristic function in (\ref{eq:characteristic}), where $g$ is a standard normal random variable independent from $S_t$.
The proof is the same as in the L\'evy case, but with time dependent parameters,  and it is obtained by direct computation  of $\mathbb{E}[e^{i \, u \, f_t }]$, conditioning w.r.t. $S_t$. 
The $f_t$ in (\ref{eq: master}) is then equivalent in law to the ATS process at maturity $t$; thus, we  can use this expression of $f_t$ to compute European options.

Consider  a European call option discounted payoff $B_t\left(F_0(t) \, e^{f_t}-F_0(t) \, e^{x}\right)^+$ 
\\(and $B_t\left(F_0(t) \, e^{x}-F_0(t) \, e^{f_t}\right)^+$ the discounted payoff for the corresponding put) 
where 
$t$ is option maturity, $K$ option strike price, \( x:=\ln{\frac{K}{F_0(t)}} \) the asset moneyness and
$B_t$ the deterministic discount factor between 0 and $t$. 
We can write the expected European call and put option price at time zero conditioning on 
$S_t$
\begin{align}
{C_t(x)} &= B_t\,F_0(t)\,
\mathbb{E}\left[\left(e^{f_t}-e^{x}\right)^+\right]\nonumber\\
&=B_t\,F_0(t)\, \mathbb{E}\left[e^{\varphi_t t-t\bar{\sigma}^2 \eta_t S_t}N\left(\frac{-x}{\bar{\sigma} \sqrt{S_t t}}+l_t^{S_t}+\frac{\bar{\sigma} \sqrt{S_t t}}{2}\right)-
e^{x}N\left(\frac{-x}{\bar{\sigma} \sqrt{S_t t}}+l_t^{S_t}-\frac{\bar{\sigma} \sqrt{S_t t}}{2}\right)\right] \label{equation::call}\\
{P_t(x)}&=B_t\,F_0(t)\,
\mathbb{E}\left[\left(e^{x}-e^{f_t}\right)^+\right]\nonumber\\
&= B_t\,F_0(t)\,
   \mathbb{E}\left[e^{x}N\left(\frac{x}{\bar{\sigma} \sqrt{S_t t}}-l_t^{S_t}+\frac{\bar{\sigma} \sqrt{S_t t}}{2}\right)-e^{\varphi_t t-t\bar{\sigma}^2 \eta_t S_t}N\left(\frac{x}{\bar{\sigma} \sqrt{S_t t}}-l_t^{S_t}-\frac{\bar{\sigma} \sqrt{S_t t}}{2}\right)\right]
   , \label{equation::put}	    
	\end{align}

	where
	\begin{equation}
     l_t^{z}:=-\bar{\sigma} \eta_t\sqrt{z\, t}+\frac{\varphi_t \sqrt{t}}{\bar{\sigma}\sqrt{z}} \label{equation::lt}
\end{equation} and $N$ is the standard normal cumulative distribution. 

\smallskip

Equations (\ref{equation::call}) and (\ref{equation::put}) are crucial in the deduction of paper's key results: let us stop and comment. 
First, let us notice that we can consider option prices with $F_0(t)=1$ and $B_t=1$ without any loss of generality: 
we are interested in the implied volatility and  these two quantities cancel out from both sides of the implied volatility equation.
Second, let us emphasize that the quantity inside the expected values are, in both equations (\ref{equation::call}) and (\ref{equation::put}), positive. 
Finally, it can be useful to mention that
a similar result has been obtained by \citet{hull1987pricing} for options on an asset with stochastic volatility. 

\smallskip


We can re-write equations (\ref{equation::call}) and (\ref{equation::put}) w.r.t to the \textit{moneyness degree} $y$  (cf. Introduction)
\begin{align*}
   { C_t(y\sqrt{t})}&=\mathbb{E}\left[e^{\varphi_t t-t\bar{\sigma}^2\eta_tS_t}N\left(-\frac{y}{{\bar{\sigma}\sqrt{S_t}}}+l_t^{S_t}+\bar{\sigma} \frac{\sqrt{S_t t}}{2}\right)-e^{y\sqrt{t}}N\left(-\frac{y}{{\bar{\sigma}\sqrt{S_t}}}+l_t^{S_t}-\bar{\sigma} \frac{\sqrt{S_t t}}{2}\right)\right]\\
  { P_t(y\sqrt{t})}&=\mathbb{E}\left[e^{y\sqrt{t}}N\left(\frac{y}{{\bar{\sigma}\sqrt{S_t}}}-l_t^{S_t}+\bar{\sigma} \frac{\sqrt{S_t t}}{2}\right)-e^{\varphi_t t-t\bar{\sigma}^2\eta_tS_t}N\left(\frac{y}{{\bar{\sigma}\sqrt{S_t}}}-l_t^{S_t}-\bar{\sigma} \frac{\sqrt{S_t t}}{2}\right)\right]\;\;, 
\end{align*}
and we can define $c_t(S_t,y)$ and $p_t(S_t,y)$ such that 
\begin{align*}
     \mathbb{E}[c_t\left(S_t,y\right)]&:={ C_t(y\sqrt{t})}\\
    \mathbb{E}[p_t\left(S_t,y\right)]&:= { P_t(y\sqrt{t})}\;\;.
\end{align*}

\citet{Black1976} option prices  w.r.t. $y$ are
\begin{align*}
   {c^{B}_t({\cal I}_t(y),y)}&=N\left(-\frac{y}{{\cal I}_t(y)}+\frac{{\cal I}_t(y)\sqrt{t}}{2}\right)-e^{y\sqrt{t}}N\left(-\frac{y}{{\cal I}_t(y)}-\frac{{\cal I}_t(y)\sqrt{t}}{2}\right)\\
      {p^{B}_t({\cal I}_t(y),y)}&=e^{y\sqrt{t}}N\left(\frac{y}{{\cal I}_t(y)}+\frac{{\cal I}_t(y)\sqrt{t}}{2}\right)-N\left(\frac{y}{{\cal I}_t(y)}-\frac{{\cal I}_t(y)\sqrt{t}}{2}\right)\;\;,\
\end{align*}
where ${\cal I}_t(y)$ is the implied volatility w.r.t. the \textit{moneyness degree}.

The implied volatility equation for the call options is \begin{equation}\label{equation::impvol}
  \mathbb{E}[c_t\left(S_t,y\right)]    =  {c^{B}_t({\cal I}_t(y),y)}\;\;
\end{equation}
and the one 
 for the put option is
\begin{equation}\label{equation::impvol_put}
    \mathbb{E}[p_t\left(S_t,y\right)]
     =   {p^{B}_t({\cal I}_t(y),y)}\;\;.
\end{equation}

In the following lemma, we prove that $\hat{\sigma}_t\sqrt{t}$ goes to zero at short-time following an approach similar to \citet[][Lemma 6.1, p.580]{alos2007short}, who considered a generalization of the Bates model. 
\begin{lemma}$ $\\ 	\label{lemma::short_time_IV}
For the ATS, at short-time, $$\hat{\sigma}_t\sqrt{t}=o(1)\;\,.$$
\end{lemma}
\begin{proof}
For an ATM put (i.e. when $y=0$), the left-hand side of equation (\ref{equation::impvol_put}) is equal to $\mathbb{E}\left[\left(1-e^{f_t}\right)^+\right]$. In the region of admissible scaling parameters, $f_t$ goes to zero in distribution because its characteristic function in (\ref{eq:characteristic}) goes to one. Hence, by the dominated convergence theorem, $\mathbb{E}\left[\left(1-e^{f_t}\right)^+\right]$ goes to zero at short time. 
For $y=0$, the right-hand side of equation (\ref{equation::impvol_put}) becomes \[
 {p^{B}_t(\hat{\sigma}_t,0)} =N\left(\frac{\hat{\sigma}_t\sqrt{t}}{2}\right)-N\left(-\frac{\hat{\sigma}_t\sqrt{t}}{2}\right)\;\;,
\]
that goes to zero if and only if $\hat{\sigma}_t\sqrt{t}$ goes to zero.
\end{proof}
Thanks to this lemma, ATM  and at short-time,
we can rewrite  the right hand side of (\ref{equation::impvol}) and (\ref{equation::impvol_put}) as \begin{equation}
    \label{equation::asymptotic_impvol}
   {c^{B}_t(\hat{\sigma}_t,0)}=   {p^{B}_t(\hat{\sigma}_t,0)} = \hat{\sigma}_t\sqrt{\frac{t}{2\pi}}+o\left(\hat{\sigma}_t\sqrt{t}\right)\;\;,
\end{equation}
where the asymptotic expansion holds because $N'(0) = \sqrt{\frac{1}{2\pi}}$, with $N'$ the standard normal probability density function.\\

\section{Short-time ATM implied volatility}

In this Section, we study the behavior of $\hat{\sigma}_t$ at short-time for the ATS. The idea of the proofs  is simple. Equation (\ref{equation::asymptotic_impvol}) is the short-time asymptotic expansion of the ATM Black call and put prices. 
We can study the short-time behavior of the ATS model price in (\ref{equation::impvol}) and (\ref{equation::impvol_put}).
\begin{enumerate}
    \item If the model price (left hand side in (\ref{equation::impvol}) and (\ref{equation::impvol_put})) goes to zero faster than $\sqrt{t}$ , then $\hat{\sigma}_0=0$ (Case 1).
        \item If the model price goes to zero slower than $\sqrt{t}$, then $\hat{\sigma}_0=\infty$ (Case 2).
    \item If the model price goes to zero as $\sqrt{t}$, then $\hat{\sigma}_0$ is finite (Cases 3, 4, 5).
\end{enumerate}
The idea of the proofs is the following. In Case 1 we bound from above the model price and we prove that it is $o\left(\sqrt{t}\right)$.  In Case 2 we bound from below the model price and we demonstrate that it goes to zero slower than $\sqrt{t}$. Finally in the remaining Cases we build upper and lower bounds on the model price and prove that both bounds are $O\left(\sqrt{t}\right)$. Furthermore, the proofs are divided in some sub-cases that correspond to particular ranges of the parameters $\beta$ and $\delta$  : we indicate with bold characters the range at the beginning of each sub-case.
\begin{proposition} $ $ \label{proposition:impvol1} \\
	For Case 1: $\begin{cases}
\beta<1\;\; \& \;\;-\min\left(\frac{1}{2},\beta\right)<\delta\leq 0\;\;\;\; or\\  \beta=\delta=0
	\end{cases} $,\\  the implied volatility is s.t. \[\hat{\sigma}_0=0\;\;.\] \end{proposition}
\begin{proof}$ $
\begin{center}
    \textbf{ $\boldsymbol{\beta<1}$ \& $\boldsymbol{-\min\left(\frac{1}{2},\beta\right)<\delta\leq 0}$ or}  $\boldsymbol{\beta=\delta=0 }$\
\end{center}
We bound $c_t\left(S_t,0\right)$ from above as follows.
\begin{align}
    c_t\left(S_t,0\right)=&N\left(l_t^{S_t}+\bar{\sigma} \frac{\sqrt{S_t t}}{2}\right)-N\left(l_t^{S_t}-\bar{\sigma}\frac{\sqrt{S_t t}}{2}\right)-\left(e^{\varphi_t t -t\bar{\sigma}^2\eta_tS_t}-1\right)N\left(l_t^{S_t}+\bar{\sigma} \frac{\sqrt{S_t t}}{2}\right)\nonumber\\
    \leq& \sqrt{\frac{t}{{2\pi}}}\;\bar{\sigma}\; \sqrt{S_t}+e^{\varphi_t t}-1\label{equation::bound_ct}\;\;.
\end{align}
In the equality we have just added and subtracted the quantity $N\left(l_t^{S_t}+\bar{\sigma} \frac{\sqrt{S_t t}}{2}\right)$.
The inequality holds because, by definition of standard normal cumulative distribution function,
\begin{equation}
N\left(l_t^{S_t}+\bar{\sigma} \frac{\sqrt{S_t t}}{2}\right)-N\left(l_t^{S_t}-\bar{\sigma}\frac{\sqrt{S_t t}}{2}\right) = \frac{1}{\sqrt{2\pi}}\int_{l_t^{S_t}-\bar{\sigma} \frac{\sqrt{S_t t}}{2}}^{l_t^{S_t}+\bar{\sigma} \frac{\sqrt{S_t t}}{2}} dz\; e^{-z^2/2}\leq \sqrt{\frac{t}{{2\pi}}}\;\bar{\sigma} \sqrt{S_t}\;\;,\label{equation::bound_N}
\end{equation} and because we bound from above the product $\left(e^{\varphi_t t -t\bar{\sigma}^2\eta_tS_t}-1\right)N\left(l_t^{S_t}+\bar{\sigma} \frac{\sqrt{S_t t}}{2}\right)$ with the (positive) maximums of both factors. \\
We bound the expected value of $c_t\left(S_t,0\right)$ as \[\mathbb{E}[c_t\left(S_t,0\right)]\leq \mathbb{E}\left[\sqrt{\frac{t}{{2\pi}}}\bar{\sigma} \sqrt{S_t}\right]+e^{\varphi_t t}-1=\sqrt{\frac{t}{{2\pi}}}\;\bar{\sigma}\;\mathbb{E}[\sqrt{S_t}]+o\left(\sqrt{t}\right) = o\left(\sqrt{t}\right) \;\;.\]

The first equality holds because $e^{\varphi_t t}-1 = O\left(\varphi_t t\right)=o\left(\sqrt{t}\right)$ and the last equality  because $\mathbb{E}[\sqrt{S_t}]$ goes to zero at short-time (see  \textbf{Lemma \ref{lemma::limit_sqrt}}). \\
 
Summarizing, the upper bound to  the ATS ATM  price in  (\ref{equation::impvol}) is $o\left(\sqrt{t}\right)$.  From (\ref{equation::asymptotic_impvol}) we have that the Black price is $O\left(\hat{\sigma}_t\sqrt{t}\right)$. 
Thus, \[
\hat{\sigma}_0 =0 \qedhere
\]\end{proof}

\begin{proposition}  \label{proposition:impvol2}$ $ \\
		For Case 2:  $-\min\left(\beta,\frac{1-\beta(1-\alpha)}{\alpha}\right)<\delta<-\frac{1}{2} \, \max\left(\beta,1\right)$, \\ $$\hat{\sigma}_0=\infty\;\;.$$ 
\end{proposition} 
\begin{proof} $ $\\
We divide the proof in two sub-cases.
\begin{center}
    \textbf{ $\boldsymbol{\beta\leq 1}$ \& $\boldsymbol{-\beta<\delta<-\frac{1}{2}}$}
\end{center}
Consider the left hand side of equation (\ref{equation::impvol}). We compute the derivative of $c_t(z,y)$ w.r.t. $z$ in $y=0$.
\begin{align}
    \frac{\partial c_t(z,0)}{\partial z}=&
    -t\bar{\sigma}^2\eta_te^{\varphi_t t -t\bar{\sigma}^2\eta_tz}N\left(l_t^{z}+\bar{\sigma} \frac{\sqrt{z t}}{2}\right)\nonumber\\&-\left(\frac{\varphi_t \sqrt{t}}{2\bar{\sigma} z^{3/2}}+\frac{\sqrt{t}\bar{\sigma}\eta_t}{2\sqrt{z}}\right)\left(e^{\varphi_t t-t\bar{\sigma}^2\eta_tz}N'\left(l_t^{z}+\bar{\sigma} \frac{\sqrt{z t}}{2}\right)-N'\left(l_t^{z}-\bar{\sigma} \frac{\sqrt{z t}}{2}\right) \right)\nonumber\\
    &+\frac{\sqrt{t}\bar{\sigma}}{4\sqrt{z}}\left(e^{\varphi_t t-t\bar{\sigma}^2\eta_tz}N'\left(l_t^{z}+\bar{\sigma} \frac{\sqrt{z t}}{2}\right)+N'\left(l_t^{z}-\bar{\sigma} \frac{\sqrt{z t}}{2}\right)\right)\label{equation::derivative}
    \;\;.
\end{align}
At short-time, for a given $z\in\left(0,\frac{\varphi_t}{\bar{\sigma}^2\eta_t }\right)$, $l_t^z =\frac{\sqrt{t}\bar{ \sigma}\eta_t}{\sqrt{z}}\left(-z+\frac{\varphi_t}{\bar{\sigma}^2\eta_t }\right)>0$. We observe that \\$e^{\varphi_t t -t\bar{\sigma}^2\eta_tz}=1+o(1)$ and $\lim_{t\to 0} l_t^z=\infty$ due to \textbf{Lemma \ref{Lemma::short_time_gamma}} point 1; then,  $N\left(l_t^{z}+\bar{\sigma} \frac{\sqrt{z t}}{2}\right)=1+o(1)$. Thus, \[ \frac{\partial c_t(z,0)}{\partial z} =-t\bar{\sigma}^2\eta_t+o(t\eta_t)\;\;,\] 
because the first term goes to zero as $t\eta_t$, while the second and the third terms go to zero as $N'(\sqrt{t}\,\eta_t)$ (i.e. as a negative exponential). 
Thus, for sufficiently small $t$,  $c_t(z,0)$ is decreasing w.r.t. $z$ in $(0,\frac{\varphi_t}{\bar{\sigma}^2\eta_t })$. We emphasize that the right extreme of the interval is increasing to one for sufficiently small $t$, see \textbf{Lemma \ref{Lemma::short_time_gamma}} points 2 and 3.\\
Fix $\tau>0$ and $S^*\in(0,\frac{\varphi_\tau}{\bar{\sigma}^2\eta_\tau })$; for any $t<\tau$
\begin{align*}
& \mathbb{E}\left[c_t(S_t,0)\right]\\
& \geq c_t(S^*,0)\mathbb{P}\left(S_t\leq S^*\right) \\
&    \geq \left\{ N\left(l_t^{S^*}+\bar{\sigma} \frac{\sqrt{S^* t}}{2}\right)-N\left(l_t^{S^*}-\bar{\sigma} \frac{\sqrt{S^* t}}{2}\right)+\left({\varphi_t t-t\bar{\sigma}^2\eta_tS^*}\right)N\left(l_t^{S^*}+\bar{\sigma} \frac{\sqrt{S^* t}}{2}\right)\right\}\mathbb{P}\left(S_t\leq S^*\right) \\
&   \geq  \left({\varphi_t t-t\bar{\sigma}^2\eta_tS^*}\right)N\left(l_t^{S^*}+\bar{\sigma} \frac{\sqrt{S^* t}}{2}\right)\mathbb{P}\left(S_t\leq S^*\right)= O(t\eta_t ) \;\;.
\end{align*}
  The first inequality  holds because $c_t(z,0)$ is positive for any $z\geq 0$ and because we bound from below the expected value with its minimum in the interval $(0,S^*)$ multiplied by the probability of the interval, $\mathbb{P}\left(S_t\leq S^*\right)$. The second inequality is due to the fact that $e^{x}\geq x+1$. Finally, the last inequality holds because, by definition of the standard normal cumulative distribution function, \begin{equation}
  N\left(l_t^{z}+\bar{\sigma} \frac{\sqrt{z t}}{2}\right)-N\left(l_t^{z}-\bar{\sigma} \frac{\sqrt{z t}}{2}\right)\geq 0\;\;,\;\;z\in \mathbb{R}^+\;\;.  \label{equation::bound_norm_cdf}
  \end{equation}
 Recall that $N\left(l_t^{S^*}+\bar{\sigma} \frac{\sqrt{S^* t}}{2}\right)=1+o(1)$; notice that $\mathbb{P}\left(S_t\leq S^*\right)$ is constant for $\beta=1$ and goes to one, by \textbf{Lemma \ref{lemma:conv_distr}} point 1, for $\beta<1$. This proves the last equality.\\
Notice that $t\,\eta_t$ goes to zero slower than $\sqrt{t}$ ($\delta<-0.5$), then the ATM call price goes to zero slower than  $\sqrt{t}$.
\begin{center}
      \textbf{$\boldsymbol{\beta> 1}$ \& $\boldsymbol{\textstyle-\frac{1-\beta(1-\alpha)}{\alpha}<\delta<-\frac{\beta}{2}}$ }
\end{center}
It exits $q$ such that $(\beta-1)/2< q<-\delta-1/2$.
We bound the ATM put price (\ref{equation::put}) from below for a sufficiently small $t$
\begin{align}
&\mathbb{E}[p_t\left(S_t,0\right)]\nonumber\\
&   \geq \mathbb{E}\left[\mathbb{1}_{S_t\geq 1+t^q}\left(N\left(-l_t^{S_t}+\bar{\sigma} \frac{\sqrt{S_t t}}{2}\right)-N\left(-l_t^{S_t}-\bar{\sigma} \frac{\sqrt{S_t t}}{2}\right)+N\left(-l_t^{S_t}-\bar{\sigma}\frac{\sqrt{S_t t}}{2}\right)\left(1-e^{\varphi_t t -t\bar{\sigma}^2\eta_tS_t}\right)\right)\right]\nonumber\\
&  \geq   \mathbb{E}\left[\mathbb{1}_{S_t\geq 1+t^q}N\left(-l_t^{S_t}-\bar{\sigma}\frac{\sqrt{S_t t}}{2}\right)\left(1-e^{\varphi_t t -t\bar{\sigma}^2\eta_tS_t}\right)\right]\nonumber\\
&  \geq \mathbb{P}(S_t\geq 1+t^q)\frac{1}{3} \left(1-e^{\varphi_t t-t\bar{\sigma}^2\eta_t(1+t^q)}\right) =: M_t\left( t^{1+q}\bar{\sigma}^2\eta_t+t  \bar{\sigma}^4 \eta^2_t k_t/2\right)\label{equation::M_t}\\ 
& \geq M_t\, t^{1+q}\bar{\sigma}^2\eta_t\;\;.\nonumber
\end{align}
 The first inequality holds because $p_t\left(S_t,0\right)$ is non negative and because we have added and subtracted the term $N\left(-l_t^{S_t}-\bar{\sigma} \frac{\sqrt{S_t t}}{2}\right)$. The second because the difference between the standard normal cumulative distribution functions is non negative, analogously to (\ref{equation::bound_norm_cdf}). The third because, for $S_t\in [1,\infty)$, $1-e^{\varphi_t t-t\bar{\sigma}^2\eta_tS_t}$  is positive and non decreasing in $S_t$; moreover, for a sufficiently small $t$, $N\left(-l_t^{S_t}-\bar{\sigma}\frac{\sqrt{S_t t}}{2}\right)> 1/3$ because \[\lim_{t\to 0 } N\left(-l_t^{z}-\bar{\sigma}\frac{\sqrt{z t}}{2}\right)\geq 1/2\;\;,\;\; z\in[1,\infty)\;\;. \] 
 The quantity $M_t$ is defined in (\ref{equation::M_t}). At short-time $M_t = 1/6+ o(1)$ because i) by \textbf{Lemma \ref{lemma::limit_distro}},  $\mathbb{P}(S_t\geq 1+t^q)$ goes to $1/2$ as $t$ goes to zero, and ii) by \textbf{Lemma \ref{Lemma::short_time_gamma}} point 1, \[ 1-e^{\varphi_t t-t\bar{\sigma}^2\eta_t(1+t^q)} =\left(t^{1+q}\bar{\sigma}^2\eta_t+t  \bar{\sigma}^4 \eta^2_t k_t/2\right)\left(1+o(1)\right) \;\;.\]
 Notice that $t^{1+q} \eta_t $ goes to zero slower than $\sqrt{t}$,  then ATM put price goes to zero slower than  $\sqrt{t}$.
\begin{center}
      \textbf{Case 2: $\boldsymbol{-\min\left(\beta,\textstyle\frac{1-\beta(1-\alpha)}{\alpha}\right)<\delta<-\max\left(\frac{\beta}{2},\frac{1}{2}\right)}$}
\end{center}
 Summing up, for both sub-cases, $\beta\leq 1$ \&  $-\beta<\delta<-1/2$ and $\beta> 1$  \& $-\frac{1-\beta(1-\alpha)}{\alpha}<\delta<-\beta/2$,  the lower bounds on the ATM option prices in (\ref{equation::impvol}) and (\ref{equation::impvol_put}) go to zero slower than $\sqrt{t}$.\\  Moreover, from (\ref{equation::asymptotic_impvol}) we have that the Black price is $O\left(\hat{\sigma}_t\sqrt{t}\right)$. Then, \[\hat{\sigma}_0 =\infty \qedhere\]

\end{proof}
\begin{proposition}  \label{proposition:impvol3}$ $ \\
 For Case 3: $\beta\geq 1$  $\&$ $\delta\geq-\beta/2$,  with the exception of the point $\left\{\beta=1,\, \delta=-1/2\right\},$\\  $$\hat{\sigma}_0\;\;\mbox{is finite}\;\;. $$
\end{proposition}
\begin{proof} $ $\\
We split the proof in three sub-cases. For each sub-case we build an upper and a lower bound, on the model price, and we demonstrate that both bounds are $O\left(\sqrt{t}\right)$ and then, that $\hat{\sigma}_0\;\;\mbox{is finite}. $
 \begin{center}
      \textbf{ $\boldsymbol{\beta> 1}$ \&$\boldsymbol{-\frac{\beta}{2}<=\delta<-\frac{1}{2}}$}
\end{center}

\textbf{Upper bound.}\\
\smallskip\\
Let us split the expected value of the ATS call in two parts 
\begin{align*}
 &\mathbb{E}[c_t\left(S_t,0\right)]\\
&=\;\mathbb{E}\left[N\left(l_t^{S_t}+\bar{\sigma} \frac{\sqrt{S_t t}}{2}\right)-N\left(l_t^{S_t}-\bar{\sigma} \frac{\sqrt{S_t t}}{2}\right)\right]+\mathbb{E}\left[N\left(l_t^{S_t}+\bar{\sigma}\frac{\sqrt{S_t t}}{2}\right)\left(e^{\varphi_t t -t\bar{\sigma}^2\eta_tS_t}-1\right)\right]\\
&=:A_1(t)+A_2(t)\;\;.
\end{align*}
We prove that both parts are bounded from above by quantities $ O\left(\sqrt{t}\right)$.
The first expected value is s.t. \begin{equation}\label{equation::expected_diff_norm}
	A_1(t) \leq \sqrt{\frac{t}{2\pi}}\;\bar{\sigma}\; \mathbb{E}[\sqrt{S_t}] = O\left(\sqrt{t}\right)\;\;,
\end{equation}
where the inequality holds true because of (\ref{equation::bound_N}) and $\sqrt{t}\; \mathbb{E}[\sqrt{S_t}]= O\left(\sqrt{t}\right)$ because, by \textbf{Lemma \ref{lemma::limit_sqrt}} point 1, $\mathbb{E}[\sqrt{S_t}]$ goes to one as $t$ goes to zero.\\
 Let us study the term $A_2(t)$.
\begin{align}
A_2(t)&<\mathbb{E}\left[ \left(e^{\varphi_t t-t \bar{\sigma}^2\eta_t S_t}-1\right)\mathbb{1}_{S_t<\varphi_t/(\bar{\sigma}^2\eta_t)}\right] \nonumber \\&= \sqrt{\frac{{t}}{{2\pi k_t}}}\int_0^{\varphi_t/(\bar{\sigma}^2\eta_t)} dz e^{-\frac{t(z-1)^2}{2k_t}} \left(e^{\varphi_t t-t\bar{\sigma}^2\eta_t  z}-1\right)\label{equation::integral_normal}\\
&\;\;\;+\int_0^{\varphi_t/(\bar{\sigma}^2\eta_t)} dz\left({\cal P}_{S_t}(z)- \sqrt{\frac{{t}}{{2\pi k_t}}}\; e^{-\frac{t(z-1)^2}{2k_t}}\right) \left(e^{\varphi_t t-t\bar{\sigma}^2\eta_t  z}-1\right)\label{equation::integral_G}\\
&\leq O\left(t^{\delta+(\beta+1)/2}\right)\nonumber\;\;,
\end{align}
where ${\cal P}_{S_t}$ is the law of $S_t$.
The first inequality is true because the quantity inside the expected value is positive on $\left(0,\frac{\varphi_t}{\bar{\sigma}^2\eta_t}\right)$ and negative elsewhere. The equality is obtained by adding and subtrancting the same expected value for a Gaussian random variable. We prove the second inequality in two steps, showing that both (\ref{equation::integral_normal}) and (\ref{equation::integral_G}) are bounded by $ O\left(t^{\delta+(\beta+1)/2}\right)$. \\ First, we consider (\ref{equation::integral_normal}) 
\begin{align}
&\sqrt{\frac{{t}}{{2\pi k_t}}}\;\int_0^{\varphi_t/(\bar{\sigma}^2\eta_t)} dz\; e^{-\frac{t(z-1)^2}{2k_t}} \left(e^{\varphi_t t-t \bar{\sigma}^2\eta_t z}-1\right)\nonumber\\=&
 \frac{1}{\sqrt{2\pi}}\int_{{\cal A}_t} dw\; e^{-\frac{w^2}{2}} \left(e^{\varphi_t t-t\bar{\sigma}^2\eta_t (1+w\sqrt{k_t/t})}-1\right) \label{equation::prop4_2}\\\leq&
\frac{1}{\sqrt{2\pi}}\int_{{\cal A}_t} dw\; e^{-\frac{w^2}{2}}e^{-\sqrt{t}\bar{\sigma}^2\eta_t \sqrt{k_t}w}-\frac{1}{\sqrt{2\pi}}\int_{{\cal A}_t} dw\; e^{-\frac{w^2}{2}}\label{equation::prop4_3}\\
=& e^{t\bar{\sigma}^4\eta_t^2k_t/2}N\left(\sqrt{t}\bar{\sigma}^2\eta_t \sqrt{k_t} +\left(\frac{\varphi_t}{\bar{\sigma}^2\eta_t}-1\right)\sqrt{\frac{t}{k_t}}\right)- N\left(\left(\frac{\varphi_t}{\bar{\sigma}^2\eta_t}-1\right)\sqrt{\frac{t}{k_t}}\right)+o\left(t^{\delta+(\beta+1)/2}\right) \label{equation::prop4_41}\\=&\sqrt{\frac{t}{2\pi}}\bar{\sigma}^2\eta_t \sqrt{k_t}+\frac{t\bar{\sigma}^4\eta_t^2k_t}{4}+o\left(t^{\delta+(\beta+1)/2}\right) =O\left(t^{\delta+(\beta+1)/2}\right)\label{equation::prop4_4}\;\;,  
\end{align}
where ${\cal A}_t\equiv \left\{w\in \mathbb{R} \,:\,  -\sqrt{\frac{t}{k_t}}<w<\left(\varphi_t/(\bar{\sigma}^2\eta_t)-1\right)\sqrt{\frac{t}{k_t}}\right\}$.
Equality (\ref{equation::prop4_2}) is due to a change of the integration variable $w:=(z-1)/\sqrt{k_t/t}$,  equality (\ref{equation::prop4_3}) to the fact that, by \textbf{Lemma \ref{Lemma::short_time_gamma}}, $e^{\varphi_t t-t\bar{\sigma}^2\eta_t}<1$. Equality (\ref{equation::prop4_41}) to a change of variable $m:=w+\sqrt{t}\bar{\sigma}^2\eta_t \sqrt{k_t}$ and to the fact that both $N\left(-\sqrt{\frac{t}{k_t}}\right)$ and $N\left(\sqrt{t}\bar{\sigma}^2\eta_t \sqrt{k_t}-\sqrt{\frac{t}{k_t}}\right)$ go to zero faster that any power of $t$. Finally, (\ref{equation::prop4_4}) holds true because of the Taylor expansion of $N$ in zero.\\

Second, we consider (\ref{equation::integral_G}) 
\begin{align*}
&\left|\int_0^{\varphi_t/(\bar{\sigma}^2\eta_t)} dz\left({\cal P}_{S_t}(z)- \sqrt{\frac{{t}}{{2\pi k_t}}}\;e^{-\frac{t(z-1)^2}{2k_t}}\right) \left(e^{\varphi_t t-t\bar{\sigma}^2\eta_t  z}-1\right)\right| \nonumber\\
\leq &\left|-\left(\mathbb{P}(S_t<0)-N\left(-\sqrt{\frac{t}{k_t}}\right)\right)\left(e^{\varphi_t t}-1\right)\right|\\&\;\;+\left|\int_0^{\varphi_t/(\bar{\sigma}^2\eta_t)}dz \left(\mathbb{P}(S_t<z)-N\left((z-1)\sqrt{\frac{t}{k_t}}\right)\right)\bar{\sigma}^2\,\eta_t\, t\,e^{\varphi_t t-t\bar{\sigma}^2\eta_t  z}\right|\\
\leq &2\frac{2-\alpha}{1-\alpha}\sqrt{\frac{k_t}{t}}\left(e^{\varphi_t t}-1\right)=   O\left(t^{\delta+(\beta+1)/2}\right)\;\;.
\end{align*}
The first inequality is due to integration by part and to the triangular inequality. The second inequality is a consequence of Jensen inequality and of \textbf{Lemma \ref{lemma::limit_distro}}.\\

\textbf{Lower bound.}\\

As discussed in the proof of \textbf{Proposition \ref{proposition:impvol2}}, for a sufficiently small $t$, $c_t\left(S_t,0\right)$ is decreasing for  $S_t\in \left(0,\frac{\varphi_t}{\bar{\sigma}^2\eta_t} \right)$ hence,\begin{align*}
 \mathbb{E}\left[c_t(S_t,0)\right] 
&\geq \mathbb{E}\left[c_t(S_t,0)\mathbb{1}_{S_t\leq \varphi_t /(\bar{\sigma}^2\eta_t)}\right] \\
&\geq c_t\left(\frac{\varphi_t }{\bar{\sigma}^2\eta_t},0\right)\mathbb{P}\left(S_t\leq \frac{\varphi_t }{\bar{\sigma}^2\eta_t} \right)= \sqrt\frac{{\varphi_t t}}{{8\pi \eta_t}}+o\left(\sqrt{t}\right)=O\left(\sqrt{t}\right)\;\;.
\end{align*}
The first inequality is because $c_t\left(S_t,0\right)$ is non negative and the second is because we bound the expected value from below with the minimum of  $c_t\left(S_t,0\right)$ multiplied by the probability of the interval  $\left(0,\frac{\varphi_t}{\bar{\sigma}^2\eta_t} \right)$.
The  equality holds  because, by \textbf{Lemma \ref{lemma::limit_distro}}, \[\lim_{t\to 0} \mathbb{P}\left(S_t\leq \frac{\varphi_t }{\bar{\sigma}^2\eta_t} \right)=\frac{1}{2}\;\;,\]
and \begin{align*}
	c_t\left(\frac{\varphi_t }{\bar{\sigma}^2\eta_t},0\right) = N\left(\sqrt{\frac{\varphi_t\, t}{4\bar{ \sigma}\eta_t}} \right)- N\left(-\sqrt{\frac{\varphi_t\, t}{4\bar{ \sigma}\eta_t}} \right) = \sqrt{\frac{\varphi_t\, t}{2\pi\bar{ \sigma}\eta_t}} +o\left(\sqrt{t}\right)\;\;,
\end{align*}
with $\sqrt{\frac{\varphi_t\, t}{8\pi\bar{ \sigma}\eta_t}}=O\left(\sqrt{t}\right) $.\\
\newpage
 
 \begin{center}
      \textbf{$\boldsymbol{\beta>1}$ \&  $\boldsymbol{\delta=-\frac{1}{2}}$}
\end{center}

\textbf{Upper bound.}\\

The upper bound on the ATS call price is the same to the one of the previous sub-case \\$-\beta/2\leq\delta<-1/2,\; \beta>1$.\\

\textbf{Lower bound.}\\

We bound the  put price from below. It exist  $H>1$ such that for a sufficiently small $t$
\begin{align*}
    \mathbb{E}\left[p_t\left(S_t,0\right)\right]&\geq    \mathbb{E}\left[p_t\left(S_t,0\right)\mathbb{1}_{S_t \in [1,H]}\right]\geq \mathbb{E}\left[\left(N\left(-l_t^{S_t}+\bar{\sigma} \frac{\sqrt{S_t t}}{2}\right)-N\left(-l_t^{S_t}-\bar{\sigma}\frac{\sqrt{S_t t}}{2}\right)\right)\mathbb{1}_{S_t\in[1,H]}\right]\\
    &\geq \mathbb{E}\left[N'\left(-l_t^{S_t}+\frac{\bar{\sigma}\sqrt{S_tt}}{2}\right)\bar{\sigma}\sqrt{S_t t}\;\mathbb{1}_{S_t\in [1,H]}\right]\\
    &\geq N'\left(\bar{\sigma}\bar{\eta}-\frac{\varphi_t\,\sqrt{t}}{\bar{\sigma}} +\frac{\bar{\sigma}\sqrt{t}}{2}\right)\bar{\sigma}\sqrt{t}\;\mathbb{P}(S_t\in[1,H]) = \sqrt{\frac{t}{8\pi}}\bar{\sigma}+o\left(\sqrt{t}\right)\;\;.
\end{align*}
The first inequality holds because $p_t(S_t,0)$ is non negative. The second because $e^{\varphi_t t -t \bar{\sigma}^2\eta_t S_t}<1$ in $[1,H]$. The third inequality is due to the fact that we bound from above the difference\[N\left(-l_t^{S_t}+\bar{\sigma} \frac{\sqrt{S_t t}}{2}\right)-N\left(-l_t^{S_t}-\bar{\sigma}\frac{\sqrt{S_t t}}{2}\right)\]  with the standard normal law evaluated in the maximum between the two  (positive) arguments multiplied by the difference of the two arguments. Notice that $\bar{\sigma}\sqrt{S_t t}$ is a positive quantity almost surely.
The last inequality holds because, by \textbf{Lemma \ref{lemma::Increasing}}, it exists $H>1$ s.t. the quantity inside the expected value is increasing  in $[1,H]$ for a sufficiently small $t$. The equality is because, by \textbf{Lemma \ref{lemma::limit_distro}}, $ \mathbb{P}(S_t\in[1,H])$ goes to $1/2$ as $t$ goes to zero and  \[
\lim_{t\to 0}N'\left(\bar{\sigma}\bar{\eta}-\frac{\varphi_t\,\sqrt{t}}{\bar{\sigma}}+\frac{\bar{\sigma}\sqrt{t}}{2}\right) = \frac{1}{\sqrt{2\pi}}\;\;.
\]

 \begin{center}
      \textbf{$\boldsymbol{\beta\geq 1}$ \& $\boldsymbol{-\frac{1}{2}<\delta\leq 0}$  }
\end{center}

\textbf{Upper bound.}\\
\smallskip\\
We can  bound $c_t(S_t,0)$ from above as in (\ref{equation::bound_ct}).\\
We bound the ATS option price as 
\begin{equation}
\mathbb{E}[c_t(S_t,0)]\leq \mathbb{E}\left[\frac{1}{\sqrt{2\pi}}\bar{\sigma} \sqrt{S_tt}\right]+e^{\varphi_t t}-1\leq  O\left(\sqrt{t}\right)\;\;.\label{equation::bound_delta_05}
\end{equation}
The last inequality holds because, by Jensen inequality with concave function $\sqrt{*}$, $\mathbb{E}[\sqrt{S_t }]\leq\sqrt{\mathbb{E}[{S_t }]} =1$ and because, by \textbf{Lemma \ref{Lemma::short_time_gamma}} point 1, $e^{\varphi_t t}-1 =  o\left(\sqrt{t}\right)$.\\

\textbf{Lower bound.}\\

 To bound $c_t\left(z,0\right)$ from below we have to study its derivative in (\ref{equation::derivative}).
Notice that, at short-time, $l_t^z = O\left(\sqrt{t}\eta_t\right)=o(1)$, due to \textbf{Lemma \ref{Lemma::short_time_gamma}} point 1, and to the fact that $\delta>-1/2$. Moreover, again due to \textbf{Lemma \ref{Lemma::short_time_gamma}} point 1, $e^{\varphi_t t -t\bar{\sigma}^2\eta_tz}=1+O\left(t\eta_t\right)$. Then, we have\\ i) The negative first term at short-time is $ o\left(\sqrt{t}\right)$\[  -t\bar{\sigma}^2\eta_te^{\varphi_t t -t\bar{\sigma}^2\eta_tz}N\left(l_t^{z}+\bar{\sigma} \frac{\sqrt{z t}}{2}\right) =O\left(t \eta_t\right)=o\left(\sqrt{t}\right)\;\;. \] ii) The second term at short-time is $o\left(\sqrt{t}\right)$ \begin{align*}
     &\left(\frac{\varphi_t \sqrt{t}}{2\bar{\sigma} z^{3/2}}+\frac{\sqrt{t}\bar{\sigma}\eta_t}{2\sqrt{z}}\right)\left(e^{\varphi_t t-t\bar{\sigma}^2\eta_tz}N'\left(l_t^{z}+\bar{\sigma} \frac{\sqrt{z t}}{2}\right)-N'\left(l_t^{z}-\bar{\sigma} \frac{\sqrt{z t}}{2}\right) \right)\\ =& O\left(\sqrt{t}\eta_t\right)\frac{e^{-(l_t^{z})^2/2-\bar{\sigma}^2z t/8}}{\sqrt{2\pi}}\left(\left(1+O(t\eta_t)\right)\left(1-\frac{l_t^{z}\bar{\sigma}\sqrt{z t}}{2}+o(t \eta_t)\right)-\left(1+\frac{l_t^{z}\bar{\sigma}\sqrt{z t}}{2}+o(t \eta_t)\right)\right)\\ = &O(\eta_t^2 t^{3/2})=o\left(\sqrt{t}\right)\;\;,
 \end{align*}
 because
  \[N'\left(l_t^{z}\pm\bar{\sigma} \frac{\sqrt{z t}}{2}\right)=e^{-(l_t^{z})^2/2-\bar{\sigma}^2z t/8}\left(1+\pm \frac{l_t^{z}\bar{\sigma}\sqrt{z t}}{2}+o(t \eta_t)\right)\;\;.\]
  iii) The positive third term at short-time is $O\left(\sqrt{t}\right) $ \[\frac{\sqrt{t}\bar{\sigma}}{4\sqrt{z}}\left(e^{\varphi_t t-t\bar{\sigma}^2\eta_tz}N'\left(l_t^{z}+\bar{\sigma} \frac{\sqrt{z t}}{2}\right)+N'\left(l_t^{z}-\bar{\sigma} \frac{\sqrt{z t}}{2}\right)\right)\\
    =\sqrt{\frac{{t}}{{8\pi \,z}}}\;\bar{\sigma}+o\left(\sqrt{t}\right)\;\;.\]
   Summarizing, the leading term in (\ref{equation::derivative}), at short-time, is the third one, which is positive. Hence, for a fixed $z>0$ and for sufficiently small $t$, $c_t(z,0)$ is increasing; thus, we can bound the expected value from below
\begin{align*}
    \mathbb{E}\left[c_t\left(S_t,0\right)\right]&\geq \mathbb{E}\left[c_t\left(S_t,0\right) \mathbb{1}_{S_t\in[1/2,3/2]}\right]>c_t\left(\frac{1}{2},0\right)\mathbb{P}\left(S_t\in \left[ \frac{1}{2},\frac{3}{2}\right]\right)\\&>\left\{N\left(l_t^{1/2}+\bar{\sigma}\sqrt{\frac{t}{8}}\right)-N\left(l_t^{1/2}-\bar{\sigma}\sqrt{\frac{t}{8}}\right)\right\}\mathbb{P}\left(S_t\in \left[ \frac{1}{2},\frac{3}{2}\right]\right) \\&>N'\left(l_t^{1/2}+\bar{\sigma}\sqrt{\frac{t}{8}}\right)\bar{\sigma}\sqrt{\frac{t}{2}}\mathbb{P}\left(S_t\in \left[ \frac{1}{2},\frac{3}{2}\right]\right)\\
    &=\left(\bar{\sigma}\sqrt{\frac{t}{{4\pi}}} +o\left(\sqrt{t}\right)\right)\mathbb{P}\left(S_t\in \left[ \frac{1}{2},\frac{3}{2}\right]\right)= O\left(\sqrt{t}\right)\;\;.
\end{align*} 
The first inequality holds because $c_t(S_t,0)$ is non negative. The second because, for a sufficiently small $t$,  $c_t(S_t,0)$ is increasing. The third is true because, for sufficiently small $t$,
 $e^{\varphi_t t-t\bar{ \sigma}^2\eta_t/2}>1$, by \textbf{Lemma \ref{Lemma::short_time_gamma}} point 3. The forth is due to the fact that the difference of the standard normal cumulative distribution functions can be bounded from below by the (positive) maximum of the two arguments multiplied by the (positive) difference of the two arguments.
 The equality is due to the fact that  $\mathbb{P}\left(S_t\in \left[ \frac{1}{2},\frac{3}{2}\right]\right)$ is constant if $\beta=1$ and goes to $1$ at short-time if $\beta>1$ because, by \textbf{Lemma \ref{lemma:conv_distr}} point 2, $S_t$ goes to one in distribution at short-time.

 \begin{center}
     \textbf{ Case 3: $\boldsymbol{\beta\geq 1}$ \& $\boldsymbol{-\frac{\beta}{2}\leq\delta\leq 0}$ $\setminus$ $\boldsymbol{\beta=1},\; \left\{\boldsymbol{\delta= -\frac{1}{2}}\right\}$}
\end{center}

Summing up, in all sub-cases the upper bound and the lower bounds of the ATS option prices in (\ref{equation::impvol}) and (\ref{equation::impvol_put}) are $O(\sqrt{t})$. Moreover, from (\ref{equation::asymptotic_impvol}) we have that the Black price is $O\left(\hat{\sigma}_t\sqrt{t}\right)$. Thus, \begin{equation*}
\hat{\sigma}_0\;\; \mbox{is finite}   \qedhere  
\end{equation*}

\end{proof}
\begin{proposition} $ $ \label{proposition:impvol45} \\
		For Cases 4 and 5: $\beta\leq 1$ $\&$  $\delta=-\frac{1}{2}$,\\
		 $$\hat{\sigma}_0\;\;\mbox{is finite}\;\;.$$
\end{proposition}
\begin{proof}$ $
 \begin{center}
      \textbf{ $\boldsymbol{\beta\leq 1}$ \&  $\boldsymbol{\delta=-\frac{1}{2}}$}
\end{center}

\textbf{Upper bound.}\\
\smallskip\\
We can  bound $c_t(S_t,0)$ from above as in (\ref{equation::bound_ct}).\\
We bound the ATS option price as 
\begin{equation}
\mathbb{E}[c_t(S_t,0)]\leq \mathbb{E}\left[\frac{1}{\sqrt{2\pi}}\bar{\sigma} \sqrt{S_tt}\right]+e^{\varphi_t t}-1= O\left(\sqrt{t}\right)\;\;.\label{equation::bound_delta_05}
\end{equation}
The equality holds because, by Jensen inequality with concave function $\sqrt{*}$, $\mathbb{E}[\sqrt{S_t }]\leq\sqrt{\mathbb{E}[{S_t }]} =1$ and because, by \textbf{Lemma \ref{Lemma::short_time_gamma}} point 1, $e^{\varphi_t t}-1 =  O\left(\sqrt{t}\right)$.\\

\textbf{Lower bound.}\\

We bound $c_t\left(S_t,0\right)$ from below as:
\begin{align*}
    c_t\left(S_t,0\right)&\geq \mathbb{1}_{S_t<\varphi_t /(\bar{\sigma}^2\eta_t )}\left(N\left(l_t^{S_t}+\bar{\sigma} \frac{\sqrt{S_t t}}{2}\right)-N\left(l_t^{S_t}-\bar{\sigma}\frac{\sqrt{S_t t}}{2}\right)+\left(\varphi_t t -t\bar{\sigma}^2\eta_tS_t\right)/2\right)\\ &\geq \mathbb{1}_{S_t<\varphi_t /(\bar{\sigma}^2\eta_t )}\left(\varphi_t t -t\bar{\sigma}^2\eta_tS_t\right)/2\;\;.
\end{align*}
The first inequality is because $c_t\left(S_t,0\right)$ is non negative, because $e^x\geq x+1$, and because the normal cumulative distribution function evaluated in a positive quantity is above 1/2. The second holds because the difference between the two normal cumulative function is non negative.\\
\begin{align*}
    \mathbb{E}[c_t\left(S_t,0\right)]&\geq \mathbb{E}\left[\mathbb{1}_{S_t<\varphi_t /(\bar{\sigma}^2\eta_t )}(\varphi_t t -t\bar{\sigma}^2\eta_tS_t)/2\right]\\&=\sqrt{t}\,\bar{\sigma}^2\,\bar{\eta}/2\; \mathbb{E}\left[\mathbb{1}_{S_t<\varphi_t /(\bar{\sigma}^2\eta_t )}(-S_t+\varphi_t /(\bar{\sigma}^2\eta_t ))\right]= O\left(\sqrt{t}\right).
\end{align*}
The last equality is due to the fact that \begin{equation}
    \mathbb{E}\left[\mathbb{1}_{S_t<\varphi_t /(\bar{\sigma}^2\eta_t )}(-S_t+\varphi_t /(\bar{\sigma}^2 \eta_t))\right]=\varphi_t /(\bar{\sigma}^2 \eta_t)\mathbb{P}(S_t<\varphi_t /(\bar{\sigma}^2\eta_t ))-\mathbb{E}[S_t\mathbb{1}_{S_t<\varphi_t /(\bar{\sigma}^2\eta_t )}] \label{equation::prop_delta05}
\end{equation}can be bounded from below with a positive constant for sufficiently small $t$. This fact can be deduced for $\beta\leq 1$.
We prove it separately for the two cases $\beta<1$ and $\beta=1$.\\ 
For $\beta<1$, let us observe that, at short-time, 
$$0\leq \mathbb{E}[S_t\mathbb{1}_{S_t<\varphi_t /(\bar{\sigma}^2\eta_t )}]\leq \mathbb{E}[S_t\mathbb{1}_{S_t<1}]=o(1)\;\;,$$ 
because, by point 2 of \textbf{Lemma \ref{Lemma::short_time_gamma}}, $\varphi_t /(\bar{\sigma}^2\eta_t)<1$ and, by definition of convergence in distribution, at short-time $\mathbb{E}[S_t\mathbb{1}_{S_t<1}]=o(1)$, because, by \textbf{Lemma \ref{lemma:conv_distr}} point 1, $S_t$ converges in distribution to 0. Moreover, at short-time, $\varphi_t /(\bar{\sigma}^2 \eta_t)\mathbb{P}(S_t<\varphi_t /(\bar{\sigma}^2\eta_t )) =1+o(1)$, by point 1 of \textbf{Lemma \ref{Lemma::short_time_gamma}} and by  point  1 of  \textbf{Lemma \ref{lemma:conv_distr}}.

For $\beta=1$, we remind that  the law of $S_t$ does not depend from $t$ and we observe that the limit of (\ref{equation::prop_delta05}) for $t$  that goes to zero is positive
\[\lim_{t\to 0}\left\{\varphi_t /(\bar{\sigma}^2 \eta_t)\mathbb{P}(S_t<\varphi_t /(\bar{\sigma}^2\eta_t ))-\mathbb{E}[S_t\mathbb{1}_{S_t<\varphi_t /(\bar{\sigma}^2\eta_t )}]\right\} =\mathbb{P}(S_t<1)-\mathbb{E}[S_t\mathbb{1}_{S_t<1}]>0\;\;,\] 
where the last inequality is due to the fact that $S_t$ has unitary mean and finite variance $\bar{k}$.\\

Summarizing, as in \textbf{Proposition \ref{proposition:impvol3}}  the upper and lower bounds of the ATM prices in (\ref{equation::impvol}) are $O\left(\sqrt{t}\right)$. 
From (\ref{equation::asymptotic_impvol}), we have that the Black price is $O\left(\hat{\sigma}_t\sqrt{t}\right)$. Thus, \begin{equation*}
\hat{\sigma}_0\;\; \mbox{is finite}   \qedhere  
\end{equation*}
\end{proof}

In the propositions above, we have proven that $\hat{\sigma}_0$ is finite only in Cases 3, 4 and 5. Only for these Cases we study the short-time skew in the next Section.

\section{Short-time skew}
In this Section, we focus on the skew term $\hat{\xi}_t$ for the ATS when $\hat{\sigma}_0$ is finite. We obtain an expression of $\hat{\xi}_t$ in \textbf{Lemma \ref{lemma::skew_formula}}  and study its short-time limit.\\
In the introduction, we have mentioned that the implied volatility skew observed in the equity market is negative and it goes to zero as one over the square root of $t$. This behavior is equivalent to a negative and finite $\hat{\xi}_0$. In this Section, we prove that   $\hat{\xi}_0$ is zero in Case 3 (\textbf{Proposition \ref{proposition::skew3}})  and is negative and finite in Cases 4 and 5 (\textbf{Proposition \ref{proposition::skew45}}). Moreover, Case 5 identifies the unique parameters' set where $\hat{\xi}_0$ can be a generic value that it is possible to calibrate from market data.

\begin{lemma} \label{lemma::skew_formula}
  The skew term $\hat{\xi}_t$ is 
\begin{equation}\label{equation::smirk}
      \hat{\xi}_t=\frac{N\left(-\frac{\hat{\sigma}_t\sqrt{t}}{2}\right)-\mathbb{E}\left[N\left(l_t^{S_t}-\bar{\sigma} \frac{\sqrt{S_t t}}{2}\right)\right]}{ N'\left(-\frac{\hat{\sigma}_t\sqrt{t}}{2}\right)}\;\;.
\end{equation}
\end{lemma}
\begin{proof}$ $
Applying  the implicit function theorem to the implied volatility equation for the call option (\ref{equation::impvol}) we obtain the derivative of the implied volatility w.r.t  $y$ 
 \[
   \frac{\partial{\cal I}_t(y)}{\partial y}=\frac{ \frac{\partial \mathbb{E}\left[c_t\left(S_t,y\right)\right]}{\partial y}-  \frac{\partial c^{B}_t\left({\cal I}_t(y),y\right)}{\partial y}}{ \frac{\partial c^{B}_t\left( {\cal I }_t(y),y\right)}{\partial  {\cal I }_t(y)}}\;\;.\displaystyle
 \]
We prove the thesis by computing the three partial derivatives separately. 

     \begin{align*}
        \;\;\;\;\;\;\;\;\quad \quad \quad \frac{\partial \mathbb{E}\left[c_t\left(S_t,y\right)\right]}{\partial y} &= -\sqrt{t} e^{\sqrt{t}\, y}\mathbb{E}\left[N\left(-\frac{y}{{\bar{\sigma}\sqrt{S_t}}}+l_t^{S_t}-\bar{\sigma} \frac{\sqrt{S_t t}}{2}\right)\right] \;\;\\
        \;\, \frac{\partial c^{B}_t\left({\cal I}_t(y),y\right)}{\partial y} &=  -\sqrt{t} e^{\sqrt{t}\, y}N\left(-\frac{y}{{\cal I}_t(y)}-\frac{{\cal I}_t(y)\sqrt{t}}{2}\right)\;\;\\
  \frac{\partial c^{B}_t\left({\cal I}_t(y),y\right)}{\partial  {\cal I }_t(y)} &= \sqrt{t}e^{\sqrt{t}\,y} N'\left(-\frac{y}{{\cal I}_t(y)}-\frac{{\cal I}_t(y)\sqrt{t}}{2}\right)\;\;. \\
\end{align*}
Notice that it is possible to exchange the expected value w.r.t. $S_t$ and the derivative w.r.t. $y$ using the Leibniz rule because the law of $S_t$ does not depend from $y$.
By substituting $y=0$ and reminding that $ {\cal I}_t(0)=\hat{\sigma}_t$, we get (\ref{equation::smirk})
\end{proof}

Notice that, because of lemma \ref{lemma::short_time_IV}, the denominator of $\hat{\xi}_t$ in (\ref{equation::smirk}), $N'\left(-\frac{\hat{\sigma}_t\sqrt{t}}{2}\right)$, goes to $\frac{1}{\sqrt{2 \pi}}$ at short-time. To study the short-time behavior of $\hat{\xi}_t$ it is sufficient to consider only the numerator of equation (\ref{equation::smirk}) \[N\left(-\frac{\hat{\sigma}_t\sqrt{t}}{2}\right)-\mathbb{E}\left[N\left(l_t^{S_t}-\bar{\sigma} \frac{\sqrt{S_t t}}{2}\right)\right] \;\;.\]
\begin{proposition}$ $\\\label{proposition::skew3}
For Case 3: $\beta\geq 1$  $\&$  $-\beta/2\leq \delta\leq 0$, with the exception of the point $\left\{\beta=1,\;\; \delta=-1/2\right\}$,  the skew term is
\[\hat{\xi}_0=0\;\;.\]
\end{proposition}
\begin{proof}$ $\\
We divide the proof in two sub-cases.
 \begin{center}
      \textbf{ $\boldsymbol{\beta= 1}$ \& $\boldsymbol{-\frac{1}{2}<\delta\leq 0}$}
\end{center}
We study the numerator of $\hat{\xi}_t$ in (\ref{equation::smirk}).
\[
\lim_{t\to 0}\left\{N\left(-\frac{\hat{\sigma}_t\sqrt{t}}{2}\right)-\mathbb{E}\left[N\left(l_t^{S_t}-\bar{\sigma}\frac{\sqrt{S_t t}}{2}\right)\right]=0\right\}\;\;.
\]
We compute the limit thanks to the dominated convergence theorem because the law of $S_t$ does not depend on $t$ and $l_t^z=o(1)$ in this sub-case.
 
\begin{center}
      \textbf{$\boldsymbol{\beta> 1}$  \& $\boldsymbol{-\frac{\beta}{2}\leq \delta\leq 0}$}
\end{center}
We want to prove that
 \begin{equation}
    \mathbb{E}\left[N\left(l_t^{S_t}-\bar{\sigma}\frac{\sqrt{S_t t}}{2}\right)\right] = \frac{1}{2}+o\left(1\right)\;\;.\label{equation::num_skew_lim}
\end{equation}
The equality holds because 
\begin{align}
       &\mathbb{E}\left[N\left(l_t^{S_t}-\bar{\sigma}\frac{\sqrt{S_t t}}{2}\right)\right]\nonumber \\=&\sqrt{\frac{{t}}{{2\pi k_t}}}\int_{0}^\infty dz\; e^{-t \frac{(z-1)^2}{2k_t}}N\left(l_t^{z}-\bar{\sigma}\frac{\sqrt{z t}}{2}\right)\label{equation::integral_skew}\\
       &+\int_{0}^\infty dz \left({\cal P}_{S_t}(z)-\sqrt{\frac{{t}}{{2\pi k_t}}}\;e^{-t \frac{(z-1)^2}{2k_t}}\right)N\left(l_t^{z}-\bar{\sigma}\frac{\sqrt{z t}}{2}\right)\label{equation::integral_skew1}\;\;,
\end{align}
where ${\cal P}_{S_t}$ is the distribution of $S_t$. We study the quantities in (\ref{equation::integral_skew}) and  (\ref{equation::integral_skew1}) separately.

First, we consider (\ref{equation::integral_skew}) 

\begin{align*}
     &\lim_{t\to 0}\sqrt{\frac{{t}}{{2\pi k_t}}}\int_{0}^\infty dz\; e^{-t \frac{(
     z-1)^2}{2k_t}}N\left(l_t^{z}-\bar{\sigma}\frac{\sqrt{z t}}{2}\right)\\ =&\lim_{t\to 0}\frac{1}{\sqrt{2\pi}}\int^\infty_{-\sqrt{\frac{k_t}{t}}}dw\; e^{- \frac{w^2}{2}}N\left(\frac{\varphi_t \sqrt{t}}{\bar{\sigma} \sqrt{1+w\sqrt{k_t/t}}}-\bar{\sigma}\eta_t\sqrt{t\left(1+w\sqrt{k_t/t}\right)}-\bar{\sigma}\frac{\sqrt{t\left(1+w\sqrt{k_t/t}\right)}}{2}\right) \\
     =&\lim_{t\to 0}\frac{1}{\sqrt{2\pi}}\int^\infty_{-\sqrt{\frac{k_t}{t}}}dw\; e^{- \frac{w^2}{2}}N\left(\bar{\sigma}\eta_t \sqrt{t}\left(1-w\sqrt{k_t/t}/2\right)-\bar{\sigma}\eta_t \sqrt{t}\left(1+w\sqrt{k_t/t}/2\right)+O\left(\sqrt{t}\right)\right) \\
     =&\frac{1}{\sqrt{2\pi}}\lim_{t\to 0} \int_{\mathbb{R}}dw\; e^{- \frac{w^2}{2}}\left\{N\left(-\bar{\sigma}\bar{\eta}\sqrt{\bar{k}}t^{\delta+\beta/2}w\right)-\frac{1}{2}+\frac{1}{2}\right\}=\frac{1}{2} \;\;.
    \end{align*}

    The first equality is obtained via a change of the integration variable $(w:=\sqrt{t}(z-1)/\sqrt{k_t})$. The second equality is due to the asymptotic of $\varphi_t t$  in \textbf{Lemma \ref{Lemma::short_time_gamma}} point 1. The third equality holds because of the dominated convergence theorem. The last is trivial because  $\left[N\left(-\bar{\eta}\sqrt{\bar{k}}t^{\delta+\beta/2}w\right)-1/2\right]$ is odd w.r.t. $w$.\\
    Second, we consider (\ref{equation::integral_skew1}) 
 \begin{align*}
 &\int_{0}^\infty dz \left({\cal P}_{S_t}(z)-\sqrt{\frac{{t}}{{2\pi k_t}}}\;e^{-t \frac{(z-1)^2}{2k_t}}\right)N\left(l_t^{z}-\bar{\sigma}\frac{\sqrt{z t}}{2}\right) \\
=&\;\left(\mathbb{P}(S_t<0))-N\left(-\sqrt{\frac{t}{k_t}}\right)\right)\\&+\int_{0}^\infty dz \left(\mathbb{P}(S_t<z)-N\left((z-1)\sqrt{\frac{t}{k_t}}\right)\right)N'\left(l_t^{z}-\bar{\sigma}\frac{\sqrt{z t}}{2}\right)\left(\frac{\varphi_t\sqrt{t}}{2\bar{\sigma}\,z^{3/2}}+\frac{\bar{\sigma} \sqrt{t}\eta_t+\bar{\sigma}\sqrt{t}/2}{2\sqrt{z}} \right)=o(1)\;\;.
 \end{align*}
The first equality is due to integration by part. The second to the fact that i)  $\mathbb{P}(S_t<0)=0$, \\ii) $N\left(-\sqrt{\frac{t}{k_t}}\right)$ go to zero as $t$ goes to zero, and iii) \begin{align*}
        &\left|\int_{0}^\infty dz\left(N\left((z-1)\sqrt{\frac{t}{k_t}}\right)-\mathbb{P}(S_t<z)\right)N'\left(l_t^{z}-\bar{\sigma}\frac{\sqrt{z t}}{2}\right)\left(\frac{\varphi_t\sqrt{t}}{2\bar{\sigma}\,z^{3/2}}+\frac{\bar{\sigma} \sqrt{t}\eta_t+\bar{\sigma}\sqrt{t}/2}{2\sqrt{z}} \right) \right|\\ \leq&
 \frac{2-\alpha}{1-\alpha} \sqrt{\frac{k_t}{t}}\int_{0}^\infty dz N'\left(l_t^{z}-\bar{\sigma}\frac{\sqrt{z t}}{2}\right)\left(\frac{\varphi_t\sqrt{t}}{2\bar{\sigma}\,z^{3/2}}+\frac{\bar{\sigma} \sqrt{t}\eta_t+\bar{\sigma}\sqrt{t}/2}{2\sqrt{z}} \right)\\
  =& \frac{2-\alpha}{1-\alpha} \sqrt{\frac{k_t}{t}}=O\left(\sqrt{\frac{k_t}{t}}\right)\;\;,
    \end{align*}
    where the inequality is due to \textbf{Lemma \ref{lemma::limit_distro}} and the first equality is due the fact that  
    
    \begin{align*}
 &\int_{0}^\infty dz N'\left(l_t^{z}-\bar{\sigma}\frac{\sqrt{z t}}{2}\right)\left(\frac{\varphi_t\sqrt{t}}{2\bar{\sigma}\,z^{3/2}}+\frac{\bar{\sigma} \sqrt{t}\eta_t+\bar{\sigma}\sqrt{t}/2}{2\sqrt{z}} \right)= -\left.N\left(l_t^{z}-\bar{\sigma}\frac{\sqrt{z t}}{2}\right)\right|_0^\infty =1\;\;.
    \end{align*}
    This proves (\ref{equation::num_skew_lim}).\\
    It is now possible to compute the short-time limit of the skew term
 \[
\lim_{t\to 0}\left\{\left(N\left(-\frac{\hat{\sigma}_t\sqrt{t}}{2}\right)-\mathbb{E}\left[N\left(l_t^{S_t}-\bar{\sigma}\frac{\sqrt{S_t t}}{2}\right)\right)\right]\right\}= 0\;\; \qedhere
\]
\end{proof}

\begin{proposition}$ $\\\label{proposition::skew45}
	For Case 4: $\beta<1$ and $\delta=-1/2$, the skew term is \[\hat{\xi}_0=-\sqrt{\frac{\pi}{2}}\;\;.\] For Case 5: $\beta=1$ and $\delta=-1/2$ the skew term is \begin{equation}
	\hat{\xi}_0 =  -\sqrt{\frac{\pi}{2}} \; \mathbb{E} \left[ erf \left(\bar{\sigma}\bar{\eta}\;  r(S_t)\right)  \right]  \label{equation::skew_delta0.5}\;\;,
	\end{equation}
	where $r(S_t) := \sqrt{2} ( 1/\sqrt{S_t}-\sqrt{S_t} )$. 
\end{proposition}
\begin{proof}$ $\\
	We prove separately the two Cases.
	\begin{center}
		\textbf{$\boldsymbol{\delta=-\frac{1}{2}}$ \& $\boldsymbol{\beta< 1}$}
	\end{center}
	Thanks to \textbf{Lemma \ref{lemma::Limit}}, the limit of the numerator of $\hat{\xi}_t$ in (\ref{equation::smirk}) can be computed simply, 
	\[\lim_{t\to 0}\left(N\left(-\frac{\hat{\sigma}_t\sqrt{t}}{2}\right)-\mathbb{E}\left[N\left(l_t^{S_t}-\bar{\sigma}\frac{\sqrt{S_t t}}{2}\right)\right]\right)=-\frac{1}{2}\;\;.\]
	Thus,
	\[
	\hat{\xi}_0 =\lim_{t \to 0}\frac{N\left(-\frac{\hat{\sigma}_t\sqrt{t}}{2}\right)-\mathbb{E}\left[N\left(l_t^{S_t}-\bar{\sigma}\frac{\sqrt{S_t t}}{2}\right)\right]}{ N'\left(-\frac{\hat{\sigma}_t\sqrt{t}}{2}\right)}=-\sqrt{\frac{\pi}{2}}\;\;.
	\]
	
	\bigskip
	
	\begin{center}
		\textbf{$\boldsymbol{\delta=-\frac{1}{2}}$ \& $\boldsymbol{\beta= 1}$}
	\end{center}
	We compute the limit in $t=0$ of the numerator of $\hat{\xi}_t$ in (\ref{equation::smirk})
	\begin{equation*}
	\lim_{t\to 0}\left(N\left(-\frac{\hat{\sigma}_t\sqrt{t}}{2}\right)-\mathbb{E}\left[N\left(l_t^{S_t}-\bar{\sigma}\frac{\sqrt{S_t t}}{2}\right)\right]\right)=\mathbb{E}\left[1/2-N\left(\bar{\sigma}\bar{\eta}  \left(1/\sqrt{S_t}-\sqrt{S_t} \right)\right)\right] \;\;.
	\end{equation*}
	We obtain the equality thanks to the dominated convergence theorem, 
because  the law of $S_t$ is constant in time.  
	We recall that $erf(z) =2N(z/\sqrt{2})-1$, 
	substituting in (\ref{equation::smirk}), we obtain (\ref{equation::skew_delta0.5})
\end{proof}

Equation (\ref{equation::skew_delta0.5}) is one of the major results of the paper. Let us stop and comment. 

First, let us notice that $\hat{\xi}_0$ in (\ref{equation::skew_delta0.5}),
is a generic function of the couple of positive parameters 
$\bar{\sigma} \bar{\eta}$ and $\bar{k}$; in 
particular the $erf$ function is odd in its argument and 
$r: \mathbb{R}^+ \to \mathbb{R}$.
Moreover $\hat{\xi}_0$ depends on the parameter
$\alpha \in [0,1)$
that selects the truncated additive process of interest. 

\smallskip

Second, 
\[
-\sqrt{\frac{\pi}{2}} \le \hat{\xi}_0 \le 0
\]
i.e. the minimum value for the skew term is $- \sqrt{\pi/{2}}$, its value in Case 4.

To show the upper bound, we can rewrite  \begin{align*}
\mathbb{E} \left[ erf \left(\bar{\sigma}\bar{\eta}\;  d(S_t)\right)  \right] &= \int_0^\infty dz\;{\cal P}_{S_t}(z) \, erf \left(\bar{\sigma}\bar{\eta}\;  r(z)\right) =  \int_0^1dz\,\left({\cal P}_{S_t}(z)-\frac{{\cal P}_{S_t}(z)}{z^2}\right) \, erf \left(\bar{\sigma}\bar{\eta}\;  r(z)\right)\,\;,
\end{align*}
where the second equality is due to the change of variable $w=1/z$, and second to $r(1/w)=-r(w)$ and to the fact that $erf(z)$ is odd. 
We also observe that 
$erf \left(\bar{\sigma}\bar{\eta}\;  r(z)\right)> 0$ in $(0,1)$. 

For the two cases where the distribution of $S_t$ is known analytically 
$\alpha=0$ (VG) and $\alpha=1/2$ (NIG), 
we can prove that the skew term $ \hat{\xi}_0 $ in (\ref{equation::skew_delta0.5}) is negative for non zero $\bar{\sigma}\bar{\eta}$ and $\bar{k}$
\citep[for the expression of the Gamma and Inverse Gaussian laws see, e.g.,][Ch.4, p.128]{Cont}.
In both cases we can prove that $\left({\cal P}_{S_t}(z)-\frac{{\cal P}_{S_t}(1/z)}{z^2}\right) > 0$ in $(0,1)$;
recall that ${\cal P}_{S_t}(z)$ does not depend from time because $\beta=1$.

In the $\alpha=0$ case, $S_t$ has the law of a Gamma random variable
\[
{\cal P}_{S_t}(z)-\frac{{\cal P}_{S_t}(1/z)}{z^2} =  {\frac {1}{\bar{k}^{1/\bar{k}}\Gamma (1/\bar{k})}} z^{1/\bar{k}} e^{-z/\bar{k}  } \left( 1- \frac{e^{-1/\bar{k} \, \left( 1/z- z \right)}}{z^{2/ \bar{k}} } \right) >0\;\;,
\]
where the inequality is true in $(0,1)$ because $1- \frac{e^{-1/\bar{k} \, \left( 1/z- z \right)}}{z^{2/ \bar{k}} }>0 $ or equivalently 
$ 1/z- z + 2 \ln z >0$. The last inequality is trivial $\forall z \in (0,1)$, 
because it is equal to zero for $z=1$ and its derivative is negative.

In the  $\alpha=1/2$ case,  $S_t$ has the law of an Inverse Gaussian random variable 
\[
{\cal P}_{S_t}(z)-\frac{{\cal P}_{S_t}(1/z)}{z^2} = \frac{1}{\sqrt{2\pi \bar{k}}}e^{-r(z)^2/(2\bar{k})}\left(\frac{1}{z^{3/2}}-\frac{1}{\sqrt{z}} \right)>0\;\;,
\] where the inequality is true because 
$\frac{1}{z^{3/2}}-\frac{1}{\sqrt{z}}>0,  \forall z \in (0,1)$.\\

In all other cases, we compute numerically the skew term $ \hat{\xi}_0 $
for different admissible values of $\bar{k}, \bar{\sigma}\bar{\eta} \in \mathbb{R}^+$  and $\alpha \in [0,1)$, 
by means of inversion of the characteristic function of $S_t$, 
showing that it is either negative or equal to zero. 
In Figure \ref{figure::Surf_expexted}, we plot the numerical estimation of the skew term for $\bar{\sigma}\bar{\eta}$ and $\bar{k}$ below 3 
(an interval in line with the situation generally observed in market data)
and for a grid of four values of $\alpha$ ($\alpha=0, 1/4, 1/2, 3/4$);
in all cases the skew term $ \hat{\xi}_0 $ 
looks rather similar: equal to zero on the boundaries ($\bar{k}=0$ and $\bar{\sigma}\bar{\eta}=0$), a negative quantity in all other cases
and a decreasing function w.r.t. both $\bar{k}$ and $\bar{\sigma}\bar{\eta}$.
In Figure \ref{figure::Section_expexted}, we plot also the skew term for 
the same four values of $\alpha$,
varying $\bar{ k}$ with $\bar{\sigma}\bar{\eta} =1$ (on the left) and 
varying $\bar{\sigma}\bar{\eta} $ for $\bar{k}=1$ (on the right):
all plots look rather similar with a decreasing $ \hat{\xi}_0 $. 
\begin{center}
	\begin{minipage}[t]{1\textwidth}
		\includegraphics[width=\textwidth]{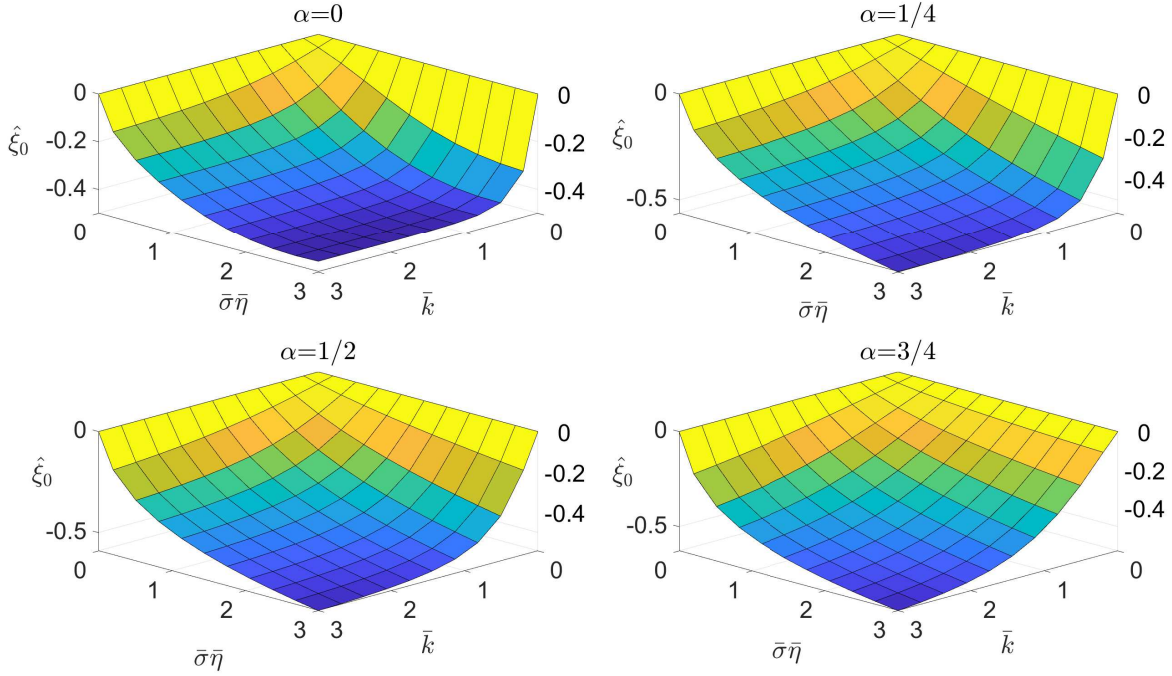}
		\captionof{figure}{\small ATS skew term  $ \hat{\xi}_0 $ for 
			$\left\{\beta=1,\, \delta=-1/2 \right\}$.
			We report $ \hat{\xi}_0 $ for four values of $\alpha$: 
			$\alpha=0$ in the upper left corner, $\alpha=1/4$ in the upper right corner, $\alpha=1/2$ in the lower left corner and $\alpha=3/4$ in the lower right corner. 
			We plot the skew  for $\bar{ k}, \bar{\sigma}\bar{\eta} \in [0,3]$. 
			In all cases the skew is negative and decreasing w.r.t. $\bar{ k}$ and $\bar{\sigma}\bar{\eta} $.}
		\label{figure::Surf_expexted}
	\end{minipage} 
\end{center}
\begin{center}
	\begin{minipage}[t]{1\textwidth}
		\includegraphics[width=\textwidth]{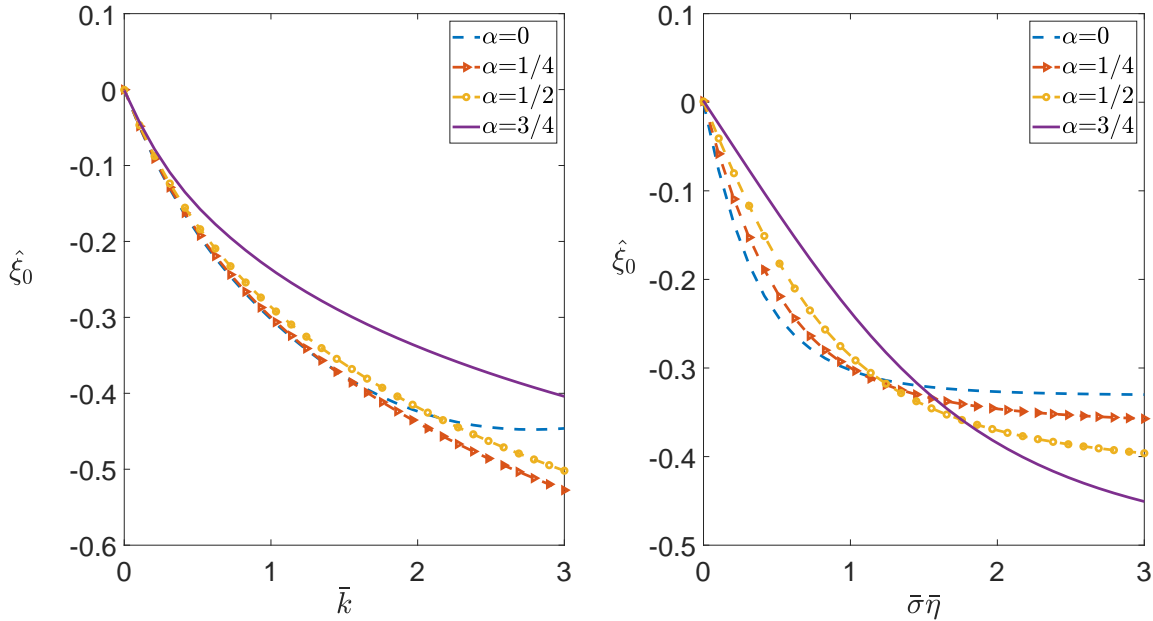}
		\captionof{figure}{\small ATS skew term  $ \hat{\xi}_0 $ for $\beta=1$ and  $\delta=-1/2$ for $\alpha=0$ (dashed blue  line), $\alpha=1/4$ (red triangles), $\alpha=1/2$ (orange circles) and, $\alpha=3/4$ (continuous violet line). 
			We plot the skew  for $\bar{ k}\in [0,3]$ with $\bar{\sigma}\bar{\eta} =1$ (on the left) and 
			for  $\bar{\sigma}\bar{\eta} \in [0,3]$ for $\bar{k}=1$ (on the right). 
			In all cases the skew is decreasing w.r.t. $\bar{ k}$ and $\bar{\sigma}\bar{\eta} $.}
		\label{figure::Section_expexted}
	\end{minipage} 
\end{center}

\smallskip

Finally, let us emphasize that the limits of $ \hat{\xi}_0 $ are zero
for $\bar{\sigma} \bar{\eta}$ and $\bar{k}$ that go to zero. 

On the one hand, recall that the law of $S_t$, ${\cal P}_{S_t}$, 
does not depend of $\bar{\sigma} \bar{\eta}$. By the dominated convergence theorem with bound ${\cal P}_{S_t}$, we have that 
\[
\lim_{\bar{\sigma}  \bar{\eta}\to 0}\mathbb{E}\left[erf \left(\bar{\sigma}\bar{\eta} \, r({S_t})\right)\right] = 0\;\;.  
\]

On the other hand, by \citet[][Th.B.9, p.308]{kijima1997markov}, we have that $S_t$ converges in distribution to $1$ as $\bar{k}$ goes to zero because \[\lim_{\bar{k}\to 0}{\cal L}_t(u;\;k_t,\;\alpha) =    \lim_{t\to 0} e^{	 \frac{1}{\bar{k}}
	\frac{1-\alpha}{\alpha}
	\left \{1-		\left(1+\frac{u \; \bar{k}}{(1-\alpha)}\right)^\alpha \right \} }=e^{-u}\,\,.\]
We are computing the expected value of a bounded function of $S_t$ that does not depend of $\bar{k}$. Thus, by definition of convergence in distribution,
\[\lim_{\bar{k}\to 0}\mathbb{E}\left[erf\left(\bar{\sigma}\bar{\eta}  \sqrt{2} \, \left(1/\sqrt{S_t}-\sqrt{S_t} \right)\right)\right] = 0\;\;.  \]

\section{Main Result}
In the following theorem, 
we present the main results of this paper. We prove that if and only if $\beta=1$ and $\delta= -\frac{1}{2}$ the ATS has a positive and constant short-time implied volatility $\hat{\sigma}_0$ and a negative and constant short-time skew $\hat{\xi}_0$. We  point out that a finite skew w.r.t. $y$ correspond to a skew that goes as $\frac{1}{\sqrt{t}}$ at short-time w.r.t. the moneyness $x$. The proof is based on the propositions of Sections 3 and 4.
\begin{theorem}\label{theorem:main_result}
	The ATS short-time implied volatility behaves as described in Table \ref{tab:Results}.
\end{theorem}
\begin{proof}
We prove that, for Case 1, $\hat{\sigma}_0=0$ in \textbf{Proposition \ref{proposition:impvol1}}.
We prove that, for Case 2, $\hat{\sigma}_0=\infty$ in \textbf{Proposition \ref{proposition:impvol2}}.
We prove that, for Cases 3, 4 and 5, $\hat{\sigma}_0$ is finite in \textbf{Proposition \ref{proposition:impvol3}} and \textbf{Proposition \ref{proposition:impvol45}}.\\ Moreover, in \textbf{Proposition \ref{proposition::skew3}} we demonstrate that, for Case 3, $\hat{\xi}_0=0$ and  in \textbf{Proposition \ref{proposition::skew45}} we demonstrate that, for Case 4, $\hat{\xi}_0=-\sqrt{\frac{\pi}{2}}$ and that, for Case 5, $\hat{\xi}_0$ is negative and finite.
\end{proof}

\section{Conclusions}

In this paper, we have analyzed the short-time behavior of the implied volatility
 of a class of pure jumps additive processes, the ATS family. 

An excellent calibration of the equity implied volatility surface has been achieved by the ATS, a class of power-law scaling additive processes \citep[see, e.g. ][]{azzone2019additive}. 
This class of processes builds upon the power-law scaling parameters $\beta$, related to the variance of jumps, and $\delta$ related to the smile asymmetry. \\

Several are the results achieved in this study.
First, we have deduced 
for this family of pure-jump additive processes the behavior of the short-time ATM implied volatility $\hat{\sigma}_t$ and the skew term $\hat{\xi}_t$  
over the region of admissible parameters (cf. \textbf{Theorem \ref{theorem:semplified_f}}). 
We get this result constructing some relevant bounds for $\hat{\sigma}_t$ and obtaining the expression of $\hat{\xi}_t$, cf. equation (\ref{equation::smirk}), via  the implicit function theorem.

Second, we have proven that
only the scaling parameters observed in empirical analysis ($\beta=1$ and $\delta=-1/2$) are compatible with the implied volatility observed in the equity market (cf. \textbf{Theorem \ref{theorem:main_result}} ).
Hence, we have demonstrated that it exists a pure-jump additive process (an exponential ATS) that, differently from the L\'evy case,  presents the two features observed in market data:
not only a finite and positive short-time implied volatility but also 
 a short-time skew proportionally inverse to the square root of the time-to-maturity. \\


\section*{Acknowledgements}
We thank Peter Carr  for an enlightening discussion on this topic. We thank E. Alòs, M. Fukasawa and all participants to the workshop New Challenges in Quantitative Finance in Barcelona. R.B. feels indebted to Peter Laurence for several helpful and wise suggestions on the subject.

\clearpage
\bibliography{sources}
\bibliographystyle{tandfx}
	\section*{Notation}
	
	\begin{center}

		\begin{tabular} {|c|l|}
			\toprule
			\textbf{Symbol}& \textbf{Description}\\ \bottomrule
			$B_t$ & discount factor between  date 0 and $t$ \\
			$c^{B}_t\left({\cal I}_t(y),y\right))$& Black call option price\\
			$c_t(S_t,t)$& quantity inside the ATS call expected value\\
			$C_t\left(x\right)$ & call option price at value date with maturity $t$ and moneyness $x$\\
			$\left\{f_t\right\}_{t\geq 0}$& sequence of random variables that models the forward exponent \\
			$F_0(t)$&  price  at time $0$ of a Forward contract with maturity $t$\\
			$g$ & standard normal random variable \\
			$I_t(x)$ & Black implied volatility with maturity $t$ and moneyness $x$ \\
			${\cal I}_t(y)$ & Black implied volatility  with maturity $t$ and {\it moneyness degree} $y$\\
			$\mathbb{1}_*$ & indicator function of the set $*$\\ 
			$k_t$ &  variance of jumps of ATS\\
			$\bar{k}$ & constant part of  variance of jumps of ATS $k_t$\\
			$K$ & option strike price\\
			$l_t^z$ & quantity defined in equation (\ref{equation::lt})\\
			$N(*)$ & standard normal cumulative distribution function evaluated in $*$\\
			$N'(*)$ & standard normal probability density function evaluated in $*$\\
			$p^{B}_t\left({\cal I}_t(y),y\right))$& Black put option price\\
			$p_t(S_t,t)$& quantity inside the ATS put expected value\\
			$P_t\left(x\right)$ & put option price at value date with maturity $t$ and moneyness $x$\\
			$\left\{S_t\right\}_{t\geq 0}$& sequence of positive random variables \\
			$t$ & time-to-maturity\\
			$x$&  option moneyness, $x=\log \frac{K}{F_0(t)}$\\
			$y$&  {\it moneyness degree}, $y:=x/\sqrt{t}$\\
			$ \alpha$ & tempered stable parameter of ATS\\
			$ \beta$& scaling parameter of $k_t$\\
			$\Gamma(*)$ & gamma function evaluated in $*$\\
			$\delta$ & scaling parameter of ${\eta}_t$\\
				$\eta_t$ &  skew parameter of ATS\\
			$\bar{\eta}$& constant part of the  skew parameter of ATS\\
			$\hat{\xi}_t$& implied volatility skew term\\
			$\hat{\xi}_0$& short-time skew term, i.e. limit for $t$ that goes to zero of $\hat{\xi}_t$\\
			$\bar{ \sigma}$ &  constant diffusion parameter of ATS\\
			$\hat{\sigma}_t$& ATM implied volatility, equal to ${\cal I}_t(0)$   \\
			$\hat{\sigma}_0$&  short-time ATM implied volatility, i.e. limit for $t$ that goes to zero of $\hat{\sigma}_t$  \\
			${\cal P}_{S_t}$ & probability density function of $S_t$\\
			$\varphi_t$&  deterministic drift term of ATS \\
		
			\bottomrule
		\end{tabular}
		\thispagestyle{plain}
	\end{center}	
\begin{appendices}

\section{Basic properties}

We report some useful results for the proofs in Section 3. In the following lemmas we consider $S_t$ of \textbf{Definition \ref{definition::S_t}} with Laplace transform ${\cal L}_t(u;\,k_t;\;\alpha)$,  at a given time $t>0$.  The proofs that follow are for the $\alpha\in (0,1)$ case. Similar proofs hold in  the $\alpha=0$ case.

\begin{lemma}
\label{theorem::Fractional_Moments}
   Let $s\in (0,1)$, then
    \begin{equation}\label{equation::fractional_moments}
        \mathbb{E}\left[S_t^s\right]=\int_{0}^{\infty}\frac{{\cal L}_t(u;\,k_t;\;\alpha)-1}{\Gamma(-s)u^{s+1}}du\;\;,
    \end{equation}
where $\Gamma$ is the Gamma function.
\end{lemma}
\begin{proof}
By elementary calculus and Fubini's Theorem \citep[see, e.g.,][Lemma 4, p.325]{urbanik1993moments}
\end{proof}
\begin{lemma}\label{theorem::Inverse_Moments}
Let $n$ be a positive integer, then \[
\mathbb{E}[S_t^{-n}]=\Gamma(n)^{-1}\int_{0}^{\infty}u^{n-1}{\cal L}_t(u;\,k_t;\;\alpha) du\;\;.
\]
 \end{lemma}
\begin{proof}
By elementary calculus and Fubini's Theorem \citep[see, e.g.,][Ch.2, p.148]{cressie1981moment}
\end{proof}

\begin{lemma}\label{lemma:bound_char}$ $
\begin{enumerate}
    \item For all $t>0$, $c\geq 1$ and $u\geq 0$ \[1-{\cal L}_t(u;\;k_t,\;\alpha)\leq 1-e^{-cu}\;\; .\]
    \item  If $\beta\geq1$,  ${\cal L}_t(u;\;k_t,\;\alpha)$ is non decreasing in $t$.
\end{enumerate}
\end{lemma}
\begin{proof}
Let us observe that
\begin{align*}
    1-{\cal L}_t(u;\;k_t,\;\alpha)&\leq 1-e^{-cu}\\
		\frac{t}{k_t}
		\displaystyle \frac{1-\alpha}{\alpha}
		\left \{		\left(1+\frac{u \; k_t}{(1-\alpha)t}\right)^\alpha-1 \right \}-cu&\leq 0\;\;.
\end{align*}
The last inequality is true for any  $c\geq 1$ and $u \geq 0$ because the left hand side is null in $u=0$ and its first order derivative w.r.t. $u$ is negative: $$\frac{1}{\left(1+\frac{k_tu}{t(1-\alpha)}\right)^{1-\alpha}}-c<0\;\;.$$\\
This proves the first point.\\

We demonstrate that the logarithm of ${\cal L}_t(u;\;k_t,\;\alpha)$ is not decreasing. Consider a positive $t$, $s\in(0,t)$ and

 \[
h(u;\;s,\;t):=\frac{t}{k_t}\left \{1-		\left(1+\frac{u \; k_t}{(1-\alpha)t}\right)^\alpha \right \}-\frac{s}{k_s}\left \{1-		\left(1+\frac{u \; k_s}{(1-\alpha)s}\right)^\alpha \right \}\;\;.
\]
We observe that $h(0;\;s,\;t)=0$ and the first order derivative  \[\displaystyle \frac{\partial h(u;\;s,\;t)}{\partial u}=\frac{1}{\left(1+\frac{k_su}{s(1-\alpha)}\right)^{1-\alpha}}-\frac{1}{\left(1+\frac{k_tu}{t(1-\alpha)}\displaystyle\right)^{1-\alpha}}\;\;\]
 is non negative $\forall u>0$ because $k_t/t$ is non decreasing in $t$, if $\beta>1$, and is constant in $t$, if $\beta=1$.
Thus, $h(u;\;s,\;t)\geq 0$, $\forall u\geq0$, and ${\cal L}_t(u;\;k_t,\;\alpha)$ is non decreasing w.r.t. $t$. This proves point 2
\end{proof}
\begin{lemma}$ $
 \begin{enumerate}
\label{lemma:conv_distr} 
\item If $\beta<1$ $S_t$ goes to zero in distribution as $t$ goes to zero.
 \item If $\beta>1$ $S_t$ goes to one in distribution as $t$ goes to zero.
 \item If $\beta=1$ the distribution of $S_t$ does not depend from $t$.
\end{enumerate}
\end{lemma}
\begin{proof}

Recall that convergence in the Laplace transform implies convergence in distribution \citep[see, e.g.,][Th.B.9, p.308]{kijima1997markov}. \\
We compute the limit of $S_t$ Laplace transform for $\beta<1$. By using the fact that $k_t/t$ goes to infinity as $t$ goes to zero we obtain
\begin{equation*}
    \lim_{t\to 0}{\cal L}_t(u;\;k_t,\;\alpha) =  \lim_{t\to 0} e^{	 \frac{t}{k_t}
		 \frac{1-\alpha}{\alpha}
		\left \{1-		\left(1+\frac{u \; k_t}{(1-\alpha)t}\right)^\alpha \right \} }=1\;\;.
\end{equation*}

 Thus, $S_t$ converges in distribution to the constant zero. This proves point 1. \\
 
 We compute the limit of $S_t$ Laplace transform for $\beta>1$. By using the fact that  $k_t/t$  goes to zero as $t$ goes to zero we obtain
\begin{equation*}
 \lim_{t\to 0}{\cal L}_t(u;\;k_t,\;\alpha) =    \lim_{t\to 0} e^{	 \frac{t}{k_t}
		 \frac{1-\alpha}{\alpha}
		\left \{1-		\left(1+\frac{u \; k_t}{(1-\alpha)t}\right)^\alpha \right \} }=e^{-u}\,\,.
\end{equation*}
 Thus, $S_t$ converges in distribution to the constant one. This proves point 2.\\
 
 Point 3 follows from the fact that, if $\beta=1$, ${\cal L}_t(u;\;k_t,\;\alpha)$ is constant in $t$

\end{proof}

\begin{lemma}
\label{lemma::limit_sqrt}
\begin{equation}
     \lim_{t\to 0} \mathbb{E}[\sqrt{S_t}] = \begin{cases}
            0 &\text{ if }\beta<1\\
            1 &\text{ if } \beta>1\\
            D &\text{ if } \beta=1
     \end{cases}\;\;,\label{equation::limit_sqrt}
\end{equation}
where $D$ is a positive constant.
\end{lemma}
\begin{proof}
  Recall that $S_t$ is a positive r.v. and $\mathbb{E}[S_t]=1$. Then, its moment of order $1/2$ is finite. 
By \textbf{Lemma \ref{theorem::Fractional_Moments}}  \[ \mathbb{E}\left[\sqrt{S_t}\right]=\int_{0}^{\infty}\frac{{\cal L}_t(u;\;k_t,\;\alpha)-1}{\Gamma(-1/2)u^{3/2}}du\;\;, \] where $\frac{-1}{\Gamma(-1/2)}\approx 3.45$. 
By \textbf{Lemma \ref{lemma:bound_char}} point 1 with $c=2$, the positive quantity  $(1-{\cal L}_t(u;\;k_t,\;\alpha))/u^{3/2}$ is lower or equal than $(1-e^{-2u})/u^{3/2}$. Thus, \begin{equation}\label{equation::bound_sqrt}
    0 \leq \mathbb{E}\left[\sqrt{S_t}\right]\leq\frac{-1}{\Gamma(-1/2)} \int^\infty_0 \frac{ 1-e^{-2u}}{u^{3/2}}du = \frac{-4}{\Gamma(-1/2)}\int_{0}^{\infty}\frac{e^{-2u}}{u^{1/2}}du= \sqrt{2}\;\;,
\end{equation} where the first equality is obtained from integration by parts and the second from the definition of $\Gamma$. Inequality (\ref{equation::bound_sqrt}) has two consequences. First, if $\beta=1$, \begin{equation}
      \lim_{t\to 0} \mathbb{E}[\sqrt{S_t}]= \mathbb{E}[\sqrt{S_t}]:=D\leq \sqrt{2}\;\;, \label{equation::limit_sqrt_beta_1}
\end{equation}  because, by \textbf{Lemma \ref{lemma:conv_distr}} point 3,  $\mathbb{E}[\sqrt{S_t}]$ is constant w.r.t. to time.
Second, we can apply the dominated convergence theorem to (\ref{equation::limit_sqrt}) for all values of $\beta$. Recall that the limits for $t$ that goes to zero of ${\cal L}_t(u;\;k_t,\;\alpha)$ for $\beta<1$ and for $\beta>1$ are computed in the proof of \textbf{Lemma \ref{lemma:conv_distr}}.\\
If $\beta<1$ \begin{align}\label{equation::limit_sqrt_beta_<1}
   \lim_{t\to 0} \mathbb{E}[\sqrt{S_t}]=\lim_{t\to 0} \frac{-1}{\Gamma(-1/2)}\int_{0}^{\infty}\frac{1-{\cal L}_t(u;\;k_t,\;\alpha)}{u^{3/2}}du =0 \;\;.
\end{align}
If $\beta>1$ 
\begin{align}
    \lim_{t\to 0} \mathbb{E}[\sqrt{S_t}]  &=\lim_{t\to 0}\frac{-1}{\Gamma(-1/2)} \int_{0}^{\infty}\frac{1-{\cal L}_t(u;\;k_t,\;\alpha)}{\Gamma(-1/2)u^{3/2}}du \nonumber\\ 
     &=\frac{-1}{\Gamma(-1/2)}\int_{0}^{\infty}  \frac{1-e^{-u}}{u^{3/2}}du=\frac{-2}{\Gamma(-1/2)}\int_{0}^{\infty}\frac{e^{-u}}{u^{1/2}}du=1\;\;, \label{equation::limit_sqrt_beta_>1}
\end{align}
 
 where the third equality is obtained from integration by parts and the third by the definition of $\Gamma$.
Equalities (\ref{equation::limit_sqrt_beta_1}), (\ref{equation::limit_sqrt_beta_<1}) and (\ref{equation::limit_sqrt_beta_>1})  prove the thesis

\end{proof}

\begin{lemma}\label{Lemma::short_time_gamma}$ $\\
Consider ${\varphi_t}$ in (\ref{equation::power_scaling_param}).  For every $\beta$ and $\delta$   in the additive process boundaries of \textbf{Theorem \ref{theorem:semplified_f}} 
 \begin{enumerate}
   \item \begin{equation} \varphi_t t=t\bar{\sigma}^2 \eta_t -t\bar{\sigma}^4\eta_t^2  k_t/2+O\left(t\eta_t^3k_t^2 \right)\;\;, \label{eq::mart_lapl} \end{equation}  where the second term $t\bar{\sigma}^4 \eta_t^2  k_t/2$ goes to zero faster than $ t\bar{\sigma}^2\eta_t$ as $t$ goes to zero. 
\item $$\frac{\varphi_t}{\bar{\sigma}^2\eta_t} \leq 1 \;\;.$$
\item $$\lim_{t\to 0} \frac{\varphi_t}{\bar{\sigma}^2\eta_t}=1\;\;,\;\;\text{for }\;\;\delta>-\min(1,\beta)\;\;. $$
 \end{enumerate}

\end{lemma}
\begin{proof}

We prove the asymptotic expansion (\ref{eq::mart_lapl}). 
In the additive process boundaries of \textbf{Theorem \ref{theorem:semplified_f}} at least either $\beta=\delta=0$ or $\delta>-\min(1,\beta)$. In the former case (\ref{eq::mart_lapl}) is trivial. In the latter, thanks to (\ref{equation::power_scaling_param}), both  $t \eta_t =t^{1+\delta}\bar{\eta}$ and $\eta_t k_t=t^{\beta+\delta}\bar{\eta}\bar{k}$ go to zero as $t$ goes to zero. Using the Taylor series expansion
\begin{align*}
    \varphi_t t&=
 \frac{t(1-\alpha)}{k_t}\left\{ \frac{ \bar{\sigma}^2\, \eta_t\, k_t}{1-\alpha}- \frac{\bar{\sigma}^4\,\eta_t^2 \, k_t^2}{2(1-\alpha)}+O\left(\eta_t^3\,k_t^3\right)\right\}\\
		&=t\,\bar{\sigma}^2\,\eta_t-t\,\bar{\sigma}^4\,\eta_t^2\,k_t/2+O\left(t\,\eta_t^3\,\kappa_t^2\right)\:\;.
\end{align*}
This proves point 1.

We prove that  $\varphi_t /(\bar{\sigma}^2\eta_t)\leq 1$. We substitute the definition of $\varphi_t $ in (\ref{equation::power_scaling_param}), for $\alpha>0$, in (\ref{eq::mart_lapl}) and we get
\begin{equation}
    \varphi_t /(\bar{\sigma}^2\eta_t)=\frac{(1-\alpha)}{\alpha \bar{\sigma}^2  \eta_t k_t  }\left(\left(1+\frac{\bar{\sigma}^2\eta_t k_t}{1-\alpha}\right)^{\alpha}-1\right)\leq 1 \;\;. \label{eq::inequalit_phi}
\end{equation}
We define $z
:=\frac{\bar{\sigma}^2\eta_t k_t}{1-\alpha}$. Then, (\ref{eq::inequalit_phi}) is equivalent to \[
	\left(1+z\right)^{\alpha}\leq 1+\alpha z\;\;, \]
which is a well known inequality. This proves point 2.\\

Point 3 is straightforward, given point 1, because, if $\delta>-\min(1,\beta)$, $\eta_t k_t$ goes to zero as $t$ goes to zero
\end{proof}

\section{Short-time limits}
\begin{lemma}\label{lemma_general_limit}
Consider a family of positive random variables $X_t$ s.t. $\lim_{t \to 0}X_t=X$ in distribution and a sequence of functions $g_t(z)\geq 0$ and uniformly bounded s.t. $\lim_{t\to 0}g_t(z) = g(z)$.\\ If $\exists\, \tau >0$ s.t. for $t\in(0,\tau)$
\begin{enumerate}[i)]
    \item $g_t(z)$ is    Lipschitz  continuous  with bounded Lipschitz constant,
    \item $|g_t(z)-g(z)|<h(z)$  with $\lim_{z\to \infty}h(z)=0$,
\end{enumerate}
 then $$\lim_{t \to 0}\mathbb{E}[g_t(X_t)]=\mathbb{E}[g(X)]\;\;.$$

\end{lemma}
\begin{proof}
     It is possible to apply the Ascoli-Arzel\'a theorem \citep[see, e.g.,,][ Th.7.25, p.158]{rudin1976principles} on every compact set $[0,K],\;K>0$, because a sequence of Lipschitz  continuous functions with bounded Lipschitz constant is equicontinous on any compact set. Thus, a sub-sequence of  $g_t(z)$ converges uniformly to $g(z)$ in any $[0,K]$.
     For every $\epsilon>0$, $\exists\,K$ s.t.
    \begin{align*}
        \lim_{t \to 0 }\mathbb{E}[\left|g_t(X_t)-g(X_t)\right|]& = \lim_{t \to 0 }\mathbb{E}\left[ \left|g_t(X_t)-g(X_t)\right|\mathbb{1}_{X_t<K}\right]+\lim_{t \to 0 }\mathbb{E}\left[\left|g_t(X_t)-g(X_t)\right|\mathbb{1}_{X_t>K}\right]<\epsilon\;\:.
    \end{align*}
    The first expected value goes to zero because $g_t(z)$ converges uniformly to $g(z)$ on $[0,K]$, as proven above via Ascoli-Arzel\'a theorem. It exists $K$ s.t. it is possible to bound the second with $\epsilon$ because $h(z)$ goes to zero as $z$ goes to infinity. \\ Moreover, $g(z)$ is bounded because it is the limit of a uniformly bounded sequence and
     \[ \lim_{t\to 0}\mathbb{E}\left[\left|g(X_t)-g(X)\right|\right] = 0\;\;.\]
     by definition of convergence in distribution, because $g(z)$ is bounded.
     We have that
    \[ 0\leq \lim_{t \to 0 } \mathbb{E}[\left|g_t(X_t)-g(X)\right|]\leq \lim_{t \to 0 }\left\{\mathbb{E}[\left|g_t(X_t)-g(X_t)\right|]+ \mathbb{E}\left[\left|g(X_t)-g(X)\right|\right]\right\}=0\;\;,\]
    this proves the thesis
    
\end{proof}

\begin{lemma}\label{lemma::Limit}
For $\delta =-1/2$, let  $X_t$ be a sequence of positive random variable s.t. $X_t\rightarrow{X}$ in distribution for $t$ that goes to zero.\\ Then,
    $$\lim_{t\to 0}\mathbb{E}\left[N\left(\bar{\sigma}\bar{\eta} \left(-\sqrt{X_t}+\varphi_t /(\bar{\sigma}^2\sqrt{X_t} \eta_t)\right)-\bar{\sigma} \sqrt{tX_t}/2\right)\right]= \mathbb{E}\left[N(\bar{\sigma}\bar{\eta}(-\sqrt{X}+1/\sqrt{X}))\right]\;\;.$$
\end{lemma}
\begin{proof}
     Define  $$g_t(z):=N\left(\bar{\sigma}\bar{\eta} \left(-\sqrt{z}+\varphi_t /(\bar{\sigma}^2\sqrt{z} \eta_t)\right)-\bar{\sigma} \sqrt{t\,z}/2\right)\;\; \mbox{and}\;\; g(z):=N(\bar{\sigma}\bar{\eta}(-\sqrt{z}+1/\sqrt{z}))\;\;.$$ We emphasize that $g_t(z)$ is uniformly bounded by one
    and $g_t(z)$ converges point-wise to $g(z)$ because, thanks to \textbf{Lemma \ref{Lemma::short_time_gamma}} point 3, $\lim_{t\to 0}\varphi_t/(\bar{\sigma}^2\eta_t)=1$. \\
We prove that  $\exists\, \tau\in(0,1)$ s.t. the derivative of $g_t(z)$ is uniformly bounded, if $t\in(0,\tau)$.  Fix $\tau\,\in(0,1)$ s.t. 
\[ 
\begin{cases} 
\displaystyle  \frac{\varphi_\tau }{ \bar{\sigma}^2\,\eta_\tau}&>\frac{2}{3}\\
\displaystyle \frac{\bar{\sigma}\bar{\eta}}{\bar{\sigma}\bar{\eta}+\bar{\sigma} \sqrt{\tau}/2}&>\frac{3}{4}\\

\end{cases}\;\; .\]
The following hold for $t<\tau$, 
\begin{align}
&\left|\frac{\partial g_t}{\partial z}\right|\nonumber\\&= N'\left(\bar{\sigma}\bar{\eta} \left(-\sqrt{z}+\varphi_t /(\bar{\sigma}^2\eta_t\, \sqrt{z} )\right)-\bar{\sigma} \sqrt{zt}/2\right)\left|\bar{\sigma}\bar{\eta}\left(-1/(2 \sqrt{z})-\varphi_t /(2\bar{\sigma}^2\eta_t\,z^{3/2} )\right)-\bar{\sigma}\sqrt{t}/(4\sqrt{z})\right|\nonumber\\ 
&= N'\left(\bar{\sigma}\bar{\eta} \left(-\sqrt{z}+\varphi_t /(\bar{\sigma}^2\eta_t\, \sqrt{z} )\right)-\bar{\sigma} \sqrt{zt}/2\right)\left( 1+\varphi_t /(\bar{\sigma}^2 \eta_t z)+\sqrt{t}/(2\bar{\eta})\right)\bar{\sigma}\bar{\eta}/\left(2\sqrt{z}\right)\nonumber\\ 
&\leq   N'\left(\bar{\sigma}\bar{\eta}\left(-\sqrt{z}+\varphi_t /(\bar{\sigma}^2\eta_t\, \sqrt{z} )\right)-\bar{\sigma} \sqrt{zt}/2\right) \left( 1+1/z +1/(2\bar{\eta})\right)\bar{\sigma}\bar{\eta}/\left(2\sqrt{z}\right)\label{eq:l1}\\
&\leq \left[ \frac{1}{\sqrt{2\pi}} \mathbb{1}_{D_2} +N'\left(\bar{\sigma}\bar{\eta} \left(-\sqrt{z}+2/(3\sqrt{z} )\right)-\tau\bar{\sigma} \sqrt{z}/2\right)\mathbb{1}_{D_1}+ N'\left(\bar{\sigma}\bar{\eta} \left(-\sqrt{z}+1/\sqrt{z} \right)\right) \mathbb{1}_{D_3}\right]\cdot\nonumber\\
&\;\;\;\left(1+1/z +1/(2\bar{\eta})\right)\bar{\sigma}\bar{\eta}/\left(2\sqrt{z}\right) := M(z)\label{eq:l2}
\;\;.
\end{align}
Inequality (\ref{eq:l1}) holds because, by \textbf{Lemma \ref{Lemma::short_time_gamma}} point 2, $\varphi_t/(\bar{\sigma}^2\eta_t)<1$ and $\tau \in (0,1)$. Let us observe that (\ref{eq:l1}) is the product of  positive quantities. In (\ref{eq:l2}) we bound from above only the first factor, the only one that still depends from $t$. Inequality (\ref{eq:l2}) is deduced by dividing the domain of $z\in \mathbb{R^+}$ in the three sets $D_1\equiv (0,1/2]$, $D_2\equiv (1/2,3/2]$ and $D_3\equiv (3/2,\infty)$.\\ For $z\in D_2$, we bound the first factor with its maximum $\frac{1}{\sqrt{2\pi}}$.\\ For $z \in D_1$, we observe that for $t<\tau$\[
\bar{\sigma}\bar{\eta} \left(-\sqrt{z}+\varphi_t /(\bar{\sigma}^2\eta_t\, \sqrt{z} )\right)-\bar{\sigma} \sqrt{zt}/2>\bar{\sigma}\bar{\eta} \left(-\sqrt{z}+2/(3\sqrt{z} )\right)-\tau\bar{\sigma} \sqrt{z}/2>0\;\;.
\]
Hence, because $N'$ is a decreasing function of its argument in $\mathbb{R}^+$,\[
N'\left(  \bar{\sigma}\bar{\eta}\left(-\sqrt{z}+\varphi_t /(\bar{\sigma}^2\eta_t\, \sqrt{z} )\right)-\bar{\sigma} \sqrt{zt}/2\right)\leq  N'\left(\bar{\sigma}\bar{\eta} \left(-\sqrt{z}+2/(3\sqrt{z} )\right)-\tau\bar{\sigma} \sqrt{z}/2\right)\;\;,\;\;z\in D_1\;\;.
\]
Finally, for $z\in D_3$
\begin{equation}\label{eq:ineq_l3}
\bar{\sigma}\bar{\eta} \left(-\sqrt{z}+\varphi_t/(\bar{\sigma}^2\eta_t\sqrt{z} )\right)-\bar{\sigma}\sqrt{t z}/2<\bar{\sigma}  \bar{\eta} \left(   -\sqrt{z}+1/\sqrt{z}\right)<0\;\;.
\end{equation}
Thus, because $N'$ is an increasing function of its argument in $\mathbb{R}^-$
\[
N'\left(  \bar{\sigma}\bar{\eta} \left(-\sqrt{z}+\varphi_t /(\bar{\sigma}^2 \eta_t\sqrt{z})\right)-\bar{\sigma}\sqrt{t z}/2\right)<N'\left(\bar{\sigma}  \bar{\eta} \left(   -\sqrt{z}+1/\sqrt{z}\right)\right)\;\;,\;\; z\in D_3\;\;.
\]
Notice that $M(z)$ is  positive and bounded on $\mathbb{R}^+$; this implies that the derivatives of $g_t(z)$ is uniformly bounded.
Thus, the sequence $g_t(z)$ is Lipschitz  continuous in $z$ with bounded Lipschitz constant on $(0,\tau)$.\\ Moreover, for $t<\tau<1$ we have that
\begin{align*}
      \left|g_t(z)-g(z)\right|& \leq \mathbb{1}_{z\in(0,1]}+N'\left(\bar{\sigma}\bar{\eta} \left(-\sqrt{z}+1/\sqrt{z} \right)\right)(\bar{\sigma}\sqrt{zt}/2+\bar{\sigma}\bar{\eta}(1-\varphi_t /(\bar{\sigma}^2\eta_t ))/\sqrt{z}) \mathbb{1}_{z\in(1,\infty)} \\
     &\leq\mathbb{1}_{z\in(0,1]}+ N'\left(\bar{\sigma}\bar{\eta} \left(-\sqrt{z}+1/\sqrt{z} \right)\right)(\bar{\sigma}\sqrt{z}/2+\bar{\sigma}\bar{\eta}/\sqrt{z}) \mathbb{1}_{z\in(1,\infty)} := h(z)\;\;.
      \end{align*}
      In the first inequality we divide the domain of $z\in \mathbb{R}^+$ in two sets, $D_1\equiv (0,1]$ and $D_2\equiv(1,\infty)$. In the first domain the difference is bounded by one. In the second set, notice that (\ref{eq:ineq_l3}) is still valid for $z>1$; then, the difference is lower than $N'$ computed on the max of the arguments of $N$ multiplied by the  positive difference of the arguments of $N$.
   The second inequality holds because $\varphi_t/(\bar{\sigma}^2\eta_t)$ is positive and $t<1$. We observe that $h(z)$ goes to zero as $z$ goes to infinity.\\
Notice that $X_t$ converges to $X$ in distribution, $g_t(z)$ is a sequence of positive function uniformly bounded,  Lipschitz  continuous  with bounded Lipschitz constant on $(0,\tau)$, and $\lim_{z\to \infty}h(z)$=0. Thus, we prove the thesis via \textbf{Lemma \ref{lemma_general_limit}}
\end{proof}
\begin{lemma}\label{lemma::limit_distro}
For $t>0$,\begin{equation}
    \sup_{z}\left|\mathbb{P}(S_t<z)-N\left((z-1)\sqrt{\frac{t}{k_t}}\right)\right|\leq \frac{2-\alpha }{1-\alpha}\sqrt{\frac{k_t}{t}}\;\;, \label{eq:distro_approx}
\end{equation}
where $S_t$  is the random variable of \textbf{Definition \ref{definition::S_t}} with Laplace transform  $ {\cal L}_t \left(u;\;k_t,\;\alpha\right)$.\\
Moreover, if $\beta>1$, \begin{enumerate}
    \item \[\lim_{t\to 0}\mathbb{P}(S_t<1)=\lim_{t\to 0}\mathbb{P}(S_t\geq 1)=\lim_{t\to 0} \mathbb{P}\left(S_t\leq \frac{\varphi_t }{\bar{\sigma}^2\eta_t} \right)=\frac{1}{2}\;\;.\]
    \item \[\lim_{t\to 0}\mathbb{P}(S_t\leq 1-t^q)=\begin{cases}
           1/2\;\; &\mbox{if}\;\;q> \frac{\beta-1}{2}\\
                      N(-1/\sqrt{\bar{k}})\;\; &\mbox{if}\;\;q=\frac{\beta-1}{2}\\
           0\;\; &\mbox{if}\;\;q< \frac{\beta-1}{2} 

    \end{cases}\;\;.\]
        \item  \[\lim_{t\to 0}\mathbb{P}(S_t\leq 1+t^q)=\begin{cases}
        	1/2\;\; &\mbox{if}\;\;q> \frac{\beta-1}{2}\\
        	N(1/\sqrt{\bar{k}})\;\; &\mbox{if}\;\;q=\frac{\beta-1}{2}\\
        	1\;\; &\mbox{if}\;\;q< \frac{\beta-1}{2} 
        	
        \end{cases}\;\;.\]

\end{enumerate}
\end{lemma}
\begin{proof}
We use an approach similar to  \citet[][Th.4.7, p.4271]{kuchler2013tempered}. Given $t>0$, $n \in \mathbb{N}$ we define $X_t^i:=S_t^i-1$ for $i=1,2,...,n$ with $S_t^i$ independent positive random variables with Laplace transform $ {\cal L}_t \left(u;\;k_t\, n,\;\alpha\right)$. The  standard deviation of  $S_t^i$ is $\Sigma_t^n:=\sqrt{k_t n/t}$. We define\\ $Q_t^n:=\sum_{i=1}^n X_t^i/(\sqrt{n}\Sigma_t^n)$.
Notice that $Q_t^n+\sqrt{n}/\Sigma_t^n$ has the same law of $\sqrt{t/k_t}S_t$ by identity in Laplace transform because 
\[\mathbb{E}\left[e^{-u\left(P^n_t+\sqrt{n}/\Sigma_t^n\right)}\right]=\mathbb{E}\left[e^{-u\sum_{i=1}^n S_t^i/(\sqrt{n}\Sigma_t^n)}\right]=e^{\sum_{i=1}^n \frac{t(1-\alpha)}{k_t
 \alpha n }\left(1-\left(1+\frac{u\sqrt{k_t}}{\sqrt{t}(1-\alpha)}\right)^\alpha\right)} = \mathbb{E}\left[e^{-u\sqrt{\frac{t}{k_t}}S_t}\right]\;\;. \]
Thus, $\sqrt{k_t/t}\;Q_t^n+1$ is equal in distribution to $S_t$.
Moreover, for any $t>0$

 \[\sup_{z}\left|\mathbb{P}(Q_t^n<z)-N(z)\right|\leq \frac{\mathbb{E}\left[\left|X_t^i\right|^3\right]}{\left(\Sigma_t^n\right)^3\sqrt{n}}<\frac{\mathbb{E}\left[(S_t^i)^3\right]+1}{(\Sigma_t^n)^3\sqrt{n}}=t^{3/2}\frac{2+3\frac{k_t n}{t}+\frac{(2-\alpha)k_t^2 n^2}{(1-\alpha)t^2}}{k_t^{3/2}n^2}=\frac{2-\alpha }{1-\alpha}\sqrt{\frac{k_t}{t}}+O\left(\frac{1}{n}\right)\;\;.\]
 The first inequality holds thanks to the Berry-Esseen theorem \citep[see, e.g.,][Th.3.4.17, p.136]{durrett2019probability}. The first equality is obtained by substituting the third moment of $S_t^i$ and in the last equality we emphasize the leading term in $1/n$.
Thus, $\forall\;\epsilon>0$ it exists  $n$  such that 
\[
\sup_{z}\left|\mathbb{P}(Q_t^n<z)-N(z)\right|<\frac{2-\alpha }{1-\alpha}\sqrt{\frac{k_t}{t}}+\epsilon \;\;.
\]
 By definition of cumulative distribution  function, we get (\ref{eq:distro_approx}).

Equation (\ref{eq:distro_approx}) allow us to prove the limits of the probability.
\begin{enumerate}
    \item \[
\mathbb{P}(S_t\leq 1) =N(0)+O\left(t^{(\beta-1)/2}\right) = 1/2+O\left(t^{(\beta-1)/2}\right)\;\;,
\] where the first equality is due to (\ref{eq:distro_approx}) and second term goes to zero because $\beta>1$. Moreover,  \[P\left(S_t\leq\frac{\varphi_t}{\bar{\sigma}^2\eta_t}\right) = N\left({\frac{\varphi_t-\bar{\sigma}^2\eta_t}{\bar{\sigma}^2\eta_t}\sqrt{\frac{t}{k_t}}}\right)+O\left(t^{(\beta-1)/2}\right)\;\;, \]
where $\lim_{t\to 0}N\left({\frac{\varphi_t-\bar{\sigma}^2\eta_t}{\bar{\sigma}^2\eta_t}\sqrt{\frac{t}{k_t}}}\right)=\frac{1}{2}$ thanks to \textbf{Lemma \ref{Lemma::short_time_gamma}} point 3 observing that \[\left({\frac{\varphi_t}{\bar{\sigma}^2\eta_t}-1}\right)\sqrt{\frac{t}{k_t}} = \sqrt{t}\bar{\sigma}^2\eta_t\sqrt{k_t} +O\left(t^{2\delta+(3\beta+1)/2}\right) = o(1)\;\;,\]
because, in the additive process boundaries, for $\beta>1$, $\delta+(\beta+1)/2>\delta+1>0$.
\item \begin{equation*}
\mathbb{P}(S_t\leq 1-t^q) =N\left(-t^q\sqrt{\frac{t}{k_t}}\right)+ O\left(t^{(\beta-1)/2}\right)\;\;,
\end{equation*}
where $N\left(-t^q\sqrt{\frac{t}{k_t}}\right)$ goes to $1/2$ if $q>(\beta-1)/2$, to $N(-1/(\sqrt{\bar{k}})$  if $q=(\beta-1)/2$, and to $0$ if $q<(\beta-1)/2$. We emphasize that the second term goes to zero as $O\left(t^{(\beta-1)/2}\right)$.
\item \begin{equation*}
	\mathbb{P}(S_t\leq 1+t^q) =N\left(t^q\sqrt{\frac{t}{k_t}}\right)+ O\left(t^{(\beta-1)/2}\right)\;\;,
\end{equation*}
where $N\left(t^q\sqrt{\frac{t}{k_t}}\right)$ goes to $1/2$ if $q>(\beta-1)/2$, to $N(1/(\sqrt{\bar{k}})$  if $q=(\beta-1)/2$, and to $1$ if $q<(\beta-1)/2$.  We emphasize that the second term goes to zero as $O\left(t^{(\beta-1)/2}\right)$.
\end{enumerate}

\end{proof}

\begin{lemma}\label{lemma::Increasing}
If $\delta=-1/2$, $\exists H>1$ s.t. \[m(z):= N'\left(-l_t^{z}+\frac{\bar{\sigma}\sqrt{zt}}{2}\right)\bar{\sigma}\sqrt{z}\;\;\;,\;\;z>0\;\;,\] is increasing  for $z\in[1,H]$ for sufficiently small $t$, where $l_t^z$ is the quantity defined in equation (\ref{equation::lt}).

\end{lemma}
\begin{proof}
We compute the derivative w.r.t. z of $m(z)$ and study its sign at short-time.
\begin{align*}
 \frac{\partial m(z)}{\partial z}  &=  \frac{\frac{\partial N'\left(-l_t^{z}+\frac{\bar{\sigma}\sqrt{zt}}{2}\right)\bar{\sigma}\sqrt{z}}{\partial z}}{ N'\left(\bar{\sigma}\bar{\eta}\left(\sqrt{z}-\frac{\varphi_t }{\bar{\sigma}^2 \eta_t\sqrt{z}}\right)+\frac{\bar{\sigma}\sqrt{tz}}{2}\right)\bar{\sigma}}\\&=  \frac{1}{2\sqrt{z}}
       -2\left(\bar{\sigma}\bar{\eta}\left(\sqrt{z}-\frac{\varphi_t}{\bar{\sigma}^2\sqrt{z} \eta_t}\right)+\frac{\bar{\sigma}\sqrt{tz}}{2}\right)\left(\bar{\sigma}\bar{\eta}\left(\frac{1}{2}+\frac{\varphi_t }{2\bar{\sigma}^2 \eta_t z}\right)+\frac{\bar{\sigma}\sqrt{t}}{4}\right)\\
      &= \frac{1}{z^{3/2}}\left(\frac{z}{2}-z^2\left(\bar{\sigma}\bar{\eta}+\bar{\sigma}\frac{\sqrt{t}}{2}\right)^2+\bar{\sigma}^2\bar{\eta}^2\frac{\varphi_t^2 }{\bar{\sigma}^4\eta_t^2}\right)\;\;.
\end{align*}
The derivative is positive if \[
0<z< \frac{1/2+\sqrt{1/4+4\left(\bar{\sigma}\bar{\eta}+\bar{\sigma}\frac{\sqrt{t}}{2}\right)^2\bar{\sigma}^2\bar{\eta}^2\frac{\varphi_t }{\bar{\sigma}^4\eta_t^2}}{}}{2\left(\bar{\sigma}\bar{\eta}+\bar{\sigma}\frac{\sqrt{t}}{2}\right)^2}\;\;.
\]
Notice that  $\exists$ $\tau$ and $H>1$ such that for every $t<\tau$
the derivative is positive if $z<H$ because for sufficiently small time \[
  \frac{1/2+\sqrt{1/4+4\left(\bar{\sigma}\bar{\eta}+\bar{\sigma}\frac{\sqrt{t}}{2}\right)^2\bar{\sigma}^2\bar{\eta}^2\frac{\varphi_t^2 }{\bar{\sigma}^4\eta_t^2}}{}}{2\left(\bar{\sigma}\bar{\eta}+\bar{\sigma}\frac{\sqrt{t}}{2}\right)^2}>\frac{1/2}{2(\bar{\sigma}\bar{\eta}+\bar{\sigma}\sqrt{t}/2)}+\frac{\bar{\eta}}{\bar{\eta}+\sqrt{t}/2}\;\frac{\varphi_t }{\bar{\sigma}^2\eta_t }>H>1\;\;,
 \]
 where the first inequality is obtained by bounding from below $1/4$ with $0$ inside the square root and the second holds because, by \textbf{Lemma \ref{Lemma::short_time_gamma}} point 3,
 \[\lim_{t\to 0} \frac{1/2}{2(\bar{\sigma}\bar{\eta}+\bar{\sigma}\sqrt{t}/2)}+\frac{\bar{\eta}}{\bar{\eta}+\sqrt{t}/2}\;\frac{\varphi_t }{\bar{\sigma}^2\eta_t } = \frac{1}{4\bar{\sigma}^4\bar{\eta}^2}+1\;\;. \]
 Thus, $m(z)$ is increasing in  [1,H] for sufficiently small $t$
\end{proof}

\end{appendices}

\end{document}